\documentclass[final]{siamltex}

\usepackage{amsfonts}
\usepackage{amsmath}
\usepackage{amssymb}
\usepackage{mathdots}
\usepackage{graphicx}
\usepackage[algosection,boxruled,linesnumbered]{algorithm2e}
\usepackage{color}

\newcommand{\mathC}{\mathbb{C}}
\newcommand{\eqnref}[1]{(\ref{#1})}
\newcommand{\sign}{{\rm sign}}
\newcommand{\spec}{{\rm spec}}

\DontPrintSemicolon
\newcommand{\rr}{\mbox{\rm r}}
\newcommand{\Span}{\mbox{span}}

\begin{document}

\title{$2$-norm error bounds and estimates for Lanczos approximations to linear systems and rational matrix functions\thanks{This work was supported by Deutsche Forschungsgemeinschaft within SFB-TR 55 ``Hadron Physics from Lattice QCD}}


\author{A.~Frommer\thanks{Fachbereich Mathematik und Naturwissenschaften, Bergische Universit\"at Wuppertal, 42097 Wuppertal, Germany \texttt{\{frommer,kkahl,rittich\}@math.uni-wuppertal.de}} 
\and K.~Kahl\footnotemark[2] \and Th.~Lippert\thanks{J\"ulich Supercomputing Centre, Forschungszentrum J\"ulich GmbH, 52425 J\"ulich, Germany \texttt{th.lippert@fz-juelich.de}} \and H.~Rittich\footnotemark[2]}

\maketitle

\begin{abstract}
The Lanczos process constructs a sequence of orthonormal vectors $v_m$ spanning a nested sequence of Krylov subspaces generated by a hermitian matrix $A$ and some starting vector $b$. In this paper we show how to cheaply recover a secondary Lanczos process starting at
an arbitrary Lanczos vector $v_m$. This secondary process is then used to efficiently obtain computable error estimates and error bounds
for the Lanczos approximations to the action of a rational matrix function on a vector.
This includes, as a special case, the Lanczos approximation to the solution of a linear system $Ax = b$.
Our approach uses the relation between the Lanczos process and quadrature as developed 
by Golub and Meurant. It is different from methods known so far because of its use of the secondary Lanczos process. 
With our approach, it is now in particular possible to efficiently obtain {\em upper bounds} for the error in the {\em 2-norm}, provided a lower bound on the smallest eigenvalue of $A$ is known. This holds in particular for a large class of rational matrix functions including best rational approximations to the inverse square root and the sign function.  We compare our approach to other existing error estimates and bounds known from the literature and include results of several numerical experiments.
\end{abstract}



\section{Introduction}
\label{sec:intro}
Our main interest in this paper is in error bounds and error estimates for Lanczos approximations
of the action of a rational matrix function on a given vector. To set the stage, we consider in this introduction the most familiar case where the rational function is $f(t) = t^ {-1}$, i.e., we consider a linear system
\begin{equation} \label{linsysdef:eq}
  A x = b
\end{equation}
where $A \in \mathbb{C}^{n \times n}$ is hermitian positive definite (hpd), large and sparse.
The method of choice to solve such a system is the conjugate gradient (CG) method of Hestenes and Stiefel \cite{hest-stiefel} in which---given an initial guess 
$x_0$---the $m$-th iterate $x_m$ is taken from the affine Krylov subspace
\[
x_0 + K_m(A,r_0), \mbox{ where } K_m(A,r_0) = \Span\{r_0,Ar_0,\ldots,A^{m-1}r_0\}
\]
such that its residual $r_m = b-Ax_m$ is orthogonal to $K_m(A,b)$. This Galerkin condition is equivalent to requiring that $x_m$
minimizes the $A$-norm of the error $x_*-x_m$, where $x_* = A^{-1}b$, over all $x \in x_0 + K_m(A,b)$. 
Algorithmically, CG is implemented using short recurrences which makes the method very efficient computationally.

In order to obtain a stopping criterion for the CG iteration it is important to have some information 
on the error $x_*-x_m$. A simple measure for the error is
the norm of the residual $r_m = b-Ax_m$, since by the definition of the $A^2$-norm
\[
\| r_m \|^2 = \langle A(x_*-x_m), A(x_*-x_m) \rangle = \| x_*-x_m \|_{A^2}^2.
\]
For linear systems, the residual is a quantity which is easily available. If the data in the matrix $A$ is not known exactly, $\| b-Ax \| \leq \epsilon$ translates into $\|b-(A+\Delta A)x\| \leq \epsilon + \| \Delta A\| \cdot \| x\|$, which bounds the residual of the perturbed matrix $A+\Delta A$, provided we know a bound on $\|\Delta A\|$. This shows that the residual is a convenient error 
measure particularly when we have inaccurate initial data.

The energy norm, i.e., the $A$-norm, of the error, $\| x_*-x_m \|_{A} = \langle x_*-x_m, A(x_*-x_m) \rangle^{1/2} $ is 
often the most natural measure for the error, since it relates to physically meaningful quantities
in many applications. Also, the 2-norm of the error, $\langle x_*-x_m, x_*-x_m \rangle^{1/2} $,
is of interest as an operator independent measure for the error, particularly in connection with rational matrix functions.  
  
For any $z \in \mathC^n$ we have
\[
\lambda_{\min} \leq \frac{\|  z \|_{A^2}}{\| z \|_2} \leq \lambda_{\max} \, \mbox{ and } \,
\lambda^{1/2}_{\min}  \leq \frac{\| z\|_{A^2}}{\| z \|_A} \leq \lambda_{\max}^{1/2},
\]
where $\lambda_{\min} $ and $\lambda_{\max}$ denote the smallest and largest eigenvalues of $A$, respectively.
Hence we have, for example,
\[
\| z \|_2 \leq \frac{1}{\lambda_{\rm min}} \|z\|_{A^2} , \enspace \| z \|_A \leq \frac{1}{\sqrt{\lambda_{\rm min}}} \|z\|_{A^2}
\]
but the factors $\frac{1}{\lambda_{\rm min}}$ and $\frac{1}{\sqrt{\lambda_{\rm min}}}$ represent
only a worst case bound; for a given $z$ the 
ratio of the norms can be substantially smaller. 
Moreover, the extremal eigenvalues or bounds for them are not 
necessarily available. 

In \cite{Meurant.97,Meurant.05}, see also \cite{golub-meurant-book} it was shown that one can enhance the CG iteration at very low computational cost to obtain, in addition to the iterates, estimates and bounds for the error of the current iterate in a retrospective manner: For a given small positive integer $k$,  error estimates for the iterate at step $m$ can be determined at step $m+k$ ($A$-norm and 2-norm); see also \cite{Strakos.Tichy.02,Strakos.Tichy.05} for an estimate for the 2-norm obtained at iteration $m+2k$. These estimates 
become more and more precise as $k$ increases. While for the $A$-norm one can obtain lower and upper bounds in this manner, one gets only a lower bound in the case of the 2-norm. We will give more details in section~\ref{comparison:sec}. 


These error estimates and bounds rely on an elegant theory relating an integral representation of the error norms with orthogonal 
polynomials, Gaussian quadrature rules and the Lanczos process, see \cite{golub-meurant93,golub-meurant97} and the book \cite{golub-meurant-book}. In the present paper we propose to use this theory in a different manner to be able to determine lower and also upper bounds for the error in the 2-norm. Actually, instead of dealing with the CG iteration for a linear system, our focus will be, more generally, on Lanczos approximations to the action $f(A)b$ of a rational matrix function $f$ on a vector $b$, in this manner continuing the work from \cite{FrSi07}. In the matrix function case, there is no natural and easily accessible ``residual'', and, as opposed to the linear system case, the $A$-norm is not a ``natural'', physically motivated measure for the error any more. We therefore focus on the 2-norm. Note that {\em upper} bounds for the error are particularly useful, since a stopping criterion based on the upper bound being less than a prescribed threshold guarantees that the actual error is indeed less than this threshold.  

\section{Lanczos process and Lanczos approximations} \label{lanczos:sec}

In this section we recall the Lanczos process (cf.\ \cite{golub-meurant-book} or \cite{Saad-book2}) and the related Lanczos approximations to vectors of the form $f(A)b$, with $f$ a function defined on the positive real axis and $b \in \mathC^n$. Note that for $f: t \to t^{-1}$ the vector $f(A)b$ is 
the solution of the linear system $A^{-1}b$. 
Assuming that $v_1 \in \mathC^{n}$ is normalized to $\|v_1\|_2 = 1$, 
the Lanczos process computes orthonormal vectors $v_1,v_2,\ldots$ such that $v_1,\ldots,v_m$ form an orthonormal basis of the nested sequence of Krylov subspaces $K_m(A,v_1)$, $m=1,2,\ldots$. Algorithmically, 
$v_{m+1}$ is obtained by orthogonalizing $Av_m$ against all previous vectors. Since $A$ is hermitian, it is actually 
sufficient to orthogonalize against $v_m$ and $v_{m-1}$, see Algorithm~\ref{algo:lanczos}. 

\begin{algorithm}
  \caption{Lanczos process}
  \label{algo:lanczos}
  choose $v_1$ such that $\| v_1 \| = 1 $ \;
  let $\beta_0 = 0$, $v_0 = 0$ \;
  \For{ $j = 1, \dots, m$ }{
    $w_j = A v_j - \beta_{j-1} v_{j-1}$ \;
    $\alpha_j = v_j^Hw_j$ \; 
    $w_j = w_j - \alpha_j v_j$ \;
    $\beta_j = \| w_j \|_2$ \;
    \lIf{ $\beta_j = 0$ }{
      stop
    } \;
    $v_{j+1} = (1/ \beta_j) \cdot w_j $
  }
\end{algorithm}

The Lanczos process can be summarized via the {\em Lanczos relation}
\begin{equation} \label{lanczos_rel:eq}
A V_m = V_{m+1}\overline{T}_{m} = V_m T_m + \beta_{m}\cdot v_{m+1}e_m^H,
\end{equation}   
where $V_m = [v_1|\ldots|v_m] \in \mathbb{C}^{n \times m}$ is the matrix containing the Lanczos vectors,
$e_m = (0,\ldots,0,1)^H \in \mathC^m$ 
and
\[
  \overline{T}_{m} = \begin{bmatrix}
    \alpha_1 & \beta_1 \\
    \beta_1 & \alpha_2 & \ddots \\
    & \ddots & \ddots & \beta_{m-1} \\
    & & \beta_{m-1} & \alpha_m \\
    & & & \beta_m
  \end{bmatrix}
  \
 = \left[  \begin{array}{cc} 
            T_{m} & \\
            \beta_{m} \cdot e_m^H  \end{array}
\right] \in \mathbb{R}^{(m+1) \times m} 
\]
with $T_{m}$ a (real) symmetric tridiagonal matrix. 

Throughout the whole paper we will use the notation $e_j$ to denote the $j$-th canonical unit vector from $\mathC^{\ell}$, where we explicitly mention the dimension $\ell$ of the space
when necessary. We just used $e_m \in \mathC^m, e_m = (0,\ldots,0,1)^H$, and we will often
use $e_1 \in \mathC^m, e_1 = (1,0,\ldots,0)^H$ etc. For ease of terminology, we will also call $\overline{T}_m$ a tridiagonal matrix, although it is not square. 

The following two basic properties of the Lanczos process will be important for this paper.

\begin{lemma} \label{invariance:lem} 
\begin{itemize}
\item[(i)] {\em Shift invariance \cite{Parlett.80}:} Let $\sigma \in \mathC$ and put $A^\sigma = A - \sigma I$. Assume that we 
start the Lanczos process for $A^\sigma$ with the same initial vector ${v}^\sigma_1 = v_1$ as for the Lanczos process for $A$. Then the matrices ${V}^\sigma_m, \overline{T}_{m}^\sigma$ of the Lanczos relation \eqnref{lanczos_rel:eq} for ${A}^\sigma$, starting with ${v}^\sigma_1$, are given by
\[
 {V}^\sigma_m = V_m, \enspace \overline{T}^\sigma_{m} = \overline{T}_{m} - \sigma \left[ \begin{array}{ccc} & I & \\
                                                                               0 & \cdots & 0 
                                               \end{array} \right].
\]
  \item[(ii)] {\em Essential uniqueness of the Lanczos relation:} Assume that we have $V_{m+1} = [V_m \mid v_{m+1}] \in \mathC^{n \times (m+1)}$ with orthonormal columns, and $\overline{T}_m \in \mathC^{(m+1) \times m}$ tridiagonal with positive off-diagonal entries, satisfying
  \begin{equation} \label{lanczos3_relation:eq}
        AV_{m} = V_{m+1} \overline{T}_m.
\end{equation}
  Then \eqnref{lanczos3_relation:eq} is the Lanczos relation for the matrix $A$ with starting vector $v_1$, i.e., the columns of $V_{m+1}$ are the Lanczos vectors and the entries of $\overline{T}_m$ 
the corresponding coefficients. 
\end{itemize}
\end{lemma}
\begin{proof} The first result follows directly by inspection of the Lanczos process, Algorithm~\ref{algo:lanczos}. For part (ii) we note that \eqnref{lanczos3_relation:eq}, together with
the assumption that $V_{m+1}$ has orthonormal columns and that $\overline{T}_m$ is tridiagonal, already implies that for $j=1,\ldots,m$ the vector $v_{j+1}$ is a positive scalar multiple of the vector $w_j$ 
which we obtain from orthogonalizing $Av_j$ against $v_j$ and $v_{j-1}$. Since $\|v_{j+1}\| = 1$, the scalar factor must be $1/\|w_j\|$, which is exactly how the Lanczos process proceeds.
\end{proof}

The $m$-th {\em Lanczos approximation} $x_m$ to the action $f(A)b$ of a matrix function $f(A)$ 
on a vector $b$ is given as
\begin{eqnarray*}
    x_m \, = \, V_mf(V_m^HAV_m)V_m^Hb \, = \, \|b\| \cdot V_mf(T_m)e_1,
\end{eqnarray*}
where the Lanczos process is started with $v_1 = (1/\|b\|)\cdot b$.
The Lanczos approximation is motivated by the fact that it is equivalent to setting $x_m = q_{m-1}(A)b$, where $q_{m-1}$ is the polynomial of degree $m-1$ which interpolates $f$ in the eigenvalues of $T_m$, i.e., the Ritz values of $A$ with respect to the subspace $K_m(A,b)$. For details, cf.\ \cite{Frommer.Simoncini.06,Higham.book08,Saad92,vdV87}. 

In the case of a linear system $Ax = b$ we want to compute $A^{-1}b$, i.e., we have $f(t) = t^{-1}$. 
The $m$-th Lanczos approximation $x_m$ is then given as 
\begin{equation} \label{lanczos_cg:eq}
    x_m = \|b \| \cdot V_m T_{m}^{-1} e_1.
\end{equation}
This is equivalent to the Galerkin condition $b-Ax_m \perp K_m(A,b)$ with $x_m \in K_m(A,b)$. Indeed, if we put $x_m = V_m y_m$ we see that $b-AV_m y_m \perp K_m(A,b)$ iff $y_m$ solves
\[
V_m^H(b-AV_my_m) = 0,
\] 
wherein $V_m^Hb = \| b \| e_1$ and, due to \eqnref{lanczos_rel:eq}, $V_m^HAV_m = T_m$. The Lanczos approximation 
$x_m$ is thus mathematically equivalent to the $m$-th iterate of the CG method with initial guess $x_0 = 0$. Note that 
if one wants to use an initial guess $x_0 \neq 0$, CG iteratively obtains corrections to $x_0$ which are the Lanczos approximations for $A^{-1}b - x_0 = A^{-1}r_0, r_0 = b-Ax_0$. 

The residuals of the CG iterates are related to the
Lanczos vectors as stated in the following lemma, cf.\ \cite{Paige.Parlett.vdVorst.95}.

\begin{lemma} \label{cgres_lanczosvec:lem}
Let $x_m$ be the $m$-th CG iterate and $r_m = b-Ax_m$ its residual. Moreover, let $v_{m+1}$ be the $m+1$-st
Lanczos vector, where the Lanczos process is started with $v_1 = (1/\|r_0\|) \cdot r_0$. Then
\[
    r_m = \rho_m \cdot v_{m+1}
\]
with
\[
\rho_m = - e_m^Hy_m \cdot \|b\| \cdot  \beta_m, \mbox{ where } y_m = T_{m}^{-1}e_1 
\in \mathC^{m}.
\]
Moreover, we have $\rho_m = (-1)^m \| r_m\|$. 
\end{lemma} 
\begin{proof} All stated results can be found in \cite{Paige.Parlett.vdVorst.95}. As an indication
for the reader we just give a short sketch for the representation of $\rho_m$ in the case $x_0 = 0$:
Using \eqnref{lanczos_rel:eq} the residual $r_m$ of the CG iterates $x_m$ from \eqnref{lanczos_cg:eq} are given as
\begin{eqnarray*}
   b-Ax_m &=& b-\|b\| \cdot AV_mT_m^{-1}e_1 \, = \, b - \|b\| \cdot V_{m+1}\overline{T}_mT_m^{-1}e_1 \\
          &=& \|b\| \cdot V_{m+1} \left( e_1 -\overline{T}_mT_m^{-1}e_1 \right) \, = \, \|b\| \cdot 
   V_{m+1} \left( e_1 -  \begin{pmatrix} I \\ 
                       \beta_m e_m^H T_m^{-1} 
                       \end{pmatrix}
                e_1 \right) \\
     &=&  - \|b\| \cdot \beta_m \cdot (e_m^H T_m^{-1} e_1) \cdot v_{m+1}.
\end{eqnarray*}
\end{proof} 

There are various ways to cheaply update the Lanczos approximation $x_m$  from 
\eqnref{lanczos_cg:eq} to $x_{m+1}$. The standard way is to update the (root-free) Cholesky factorization
of $T_m$ to one of $T_{m+1}$, thus arriving at the familiar coupled two-term recurrence
of the CG algorithm; see \cite{Saad-book2}, e.g. Another possibility is to use the fact that $v_m = \hat{p}_m(A)b$ where $\hat{p}_m$ is the characteristic polynomial of $T_m$. The Lanczos relation \eqnref{lanczos_rel:eq} gives a 
three-term recurrence for $\hat{p}_m$. By Lemma~\ref{cgres_lanczosvec:lem}, we have 
$r_m = p_m(A)b$ with
$p_m(t) = \rho_m\hat{p}_m(t)$, $\rho_m = 1/\hat{p}_m(0)$. Since $x_m = q_{m-1}(A)b$ with $p_m(t) = 1 - tq_{m-1}(t)$, the recurrence 
for the $p_m$ implies 
one for the iterates $x_m$. Note that $\hat{p}_m(0) \neq 0$, since the zeros of $p_m$ are the 
eigenvalues of $T_m$ and thus contained in $[\lambda_{\min},\lambda_{\max}]$. We refer to \cite{Saad-book2} for a more detailed description of this approach. For future reference, this three-term recurrence variant of the CG method is given in Algorithm~\ref{algo:cg-lanczos}. 

\begin{algorithm}
  \caption{CG Lanczos (initial guess is zero)}
  \label{algo:cg-lanczos}
  set $x_{-1} = 0$, $\rho_0 = \| b \| $, $ \tau_0 = 1$, $v_1 = (1/\rho_0)b$ \;
  \For{ $j = 0, 1, \dots$ }{

    compute $\alpha_{j+1}$, $\beta_{j+1}$, $v_{j+2}$ using the Lanczos process for $A$

    \If{ $j > 0 $ }{
      $\tau_j = \left[ 
        1 - 
        \frac{ \alpha_j }{ \alpha_{j+1} }
        \frac{ \rho_j^2 }{ \rho_{j-1}^2 }
        \frac{ 1 }{ \tau_{j-1} }
      \right]^{-1} $
    }
    $ \rho_{j+1} 
    = - \tau_j \rho_j \tfrac{ \beta_{j+1} }{ \alpha_{j+1} }$\;
    $ x_{j+1} = \tau_j (x_j + \tfrac{1}{\alpha_{j+1}} r_j) + (1 - \tau_j) x_{j-1} $ \;
    $ r_{j+1} = \rho_{j+1} v_{j+2} $ \;
  }
\end{algorithm}

In our context, the major advantage of Algorithm~\ref{algo:cg-lanczos} is that it easily also produces the Lanczos approximations for systems of the form $(A-\sigma I)x = b$ if $\sigma \not \in [\lambda_{\min},\lambda_{\max}]$ and thus, in particular, if $\sigma$ is not real.
Indeed, as was observed in \cite{Fro03,FroMa99}, e.g., due to Lemma~\ref{invariance:lem}
the characteristic polynomial $\hat{p}_m^{\sigma}$ for the shifted system is related to that of the non-shifted system via $\hat{p}_m^{\sigma}(t) = \hat{p}_m(t-\sigma)$. Since $\hat{p}_m^{\sigma}(0) 
= \hat{p}_m(-\sigma) \neq 0$, we see that all Lanczos approximations are well-defined 
and that we can work out the three term recurrence for the Lanczos approximations in exactly the same manner as in the case without the shift $\sigma$. We refer to \cite{freu90} to yet another breakdown free variant, based on a short recurrence update for QR-factorizations of the matrices $T_m$. For the case of real shifts and the standard coupled two-term recurrence, see also \cite{vdEshof.Sleijpen.04ANM}.

Now, let $f: t \to \sum_{i=1}^p \frac{\omega_i}{t-\sigma_i}$ be a rational function with poles $\sigma_i$ outside
the interval $[\lambda_{\rm min},\lambda_{\rm max}]$. Complex poles $\sigma_i$ arise quite naturally in applications, such as rational approximations to the exponential function. 
Let $v_1 = (1/\|b\|)\cdot b$ be the normalized vector for $b$ 
with which we start the Lanczos process. From the shift invariance, Lemma~\ref{invariance:lem}(i), it follows
that the $m$-th Lanczos approximation $x_m$ to $f(A)b$ is given as 
\begin{equation} \label{lanczos_rat:eq}
x_m = \|b \| \cdot V_m \cdot \left( \sum_{i=1}^p \omega_i (T_m-\sigma_iI)^{-1} e_1 \right).
\end{equation}
This Lanczos approximation always exists, i.e., all matrices $T_m - \sigma_i I$ are non-singular, since the spectrum of $T_m$ is contained in 
$[\lambda_{\rm min},\lambda_{\rm max}]$ and $\sigma_i \not \in [\lambda_{\rm min},\lambda_{\rm max}]$.
From the shift invariance property and relation \eqnref{lanczos_cg:eq}, 
we see that the Lanczos iterate $x_m$ from \eqnref{lanczos_rat:eq} is just the linear combination 
\begin{equation} \label{lincomb:eq}
x_m = \sum_{i=1}^p \omega_i x_m^{(i)}
\end{equation} 
of the Lanczos approximations $x_m^{(i)}$ for the solutions of the systems $(A-\sigma_i I)x = b$ 
(with initial guess $x_0^{(i)} = 0$ for all $i$). 

By the preceeding discussion, all Lanczos approximations can be obtained via Algorithm~\ref{algo:cg-lanczos}; and by Lemma~\ref{invariance:lem}(i) we get 
the same Lanczos vectors $v_j$, independently of the shift $\sigma_i$. 
Hence we can modify Algorithm~\ref{algo:cg-lanczos} to a {\em multishift} variant, where 
we perform lines 4 to 9 simultaneously for each shift $\sigma_i$ to obtain all $p$ Lanczos approximations $x_m^{(i)}$ 
(and their linear combination $x_m$) using short recurrences and just one matrix-vector multiplication per step. 

The task to which this paper is devoted is to obtain good error estimates for the Lanczos iterates
for a single system $Ax=b$ (see \eqnref{lanczos_cg:eq}), 
or the action of a rational matrix function $f(A)b$ (the Lanczos 
iterates from \eqnref{lanczos_rat:eq}). In the case of the CG iterates for the system $Ax =b$, we can express the error as
\[
  x_*-x_m  = A^{-1}r_m \mbox{ with } r_m = b-Ax_m,
\]
which, using Lemma~\ref{cgres_lanczosvec:lem} results in
\begin{equation} \label{linsys_error:eq}
\| x_*-x_m   \|_2^2 = |\rho_m|^2 \cdot v_{m+1}^HA^{-2}v_{m+1}, \enspace  \| x_*-x_m \|_A^2 = |\rho_m|^2 \cdot v_{m+1}^HA^{-1}v_{m+1}.
\end{equation}

For the Lanczos approximation \eqnref{lanczos_rat:eq} for a rational function 
we can apply 
Lemma~\ref{cgres_lanczosvec:lem} to all systems $(A-\sigma_i I)x^{(i)} = b$ to see that we have
\[
r_m^{(i)} = b - (A-\sigma_i I)x_m^{(i)} = \rho_m^{(i)} v_{m+1}, 
\]
so that it is possible to express the error $f(A)b - x_m =  \sum_{i=1}^p \omega_i (A-\sigma_i I)^{-1} b - x_m$ with
$x_m$ from \eqnref{lincomb:eq} as
\begin{eqnarray*}
\sum_{i=1}^p \omega_i (A-\sigma_i I)^{-1} b - \omega_i x_m^{(i)} & = & \sum_{i=1}^p \omega_i (A-\sigma_i I)^{-1} \left(b-(A-\sigma_i I)x_m^{(i)}\right) \\
&=& \sum_{i=1}^p \omega_i \rho_m^{(i)}(A-\sigma_i I)^{-1} v_{m+1} \\
& = & g_m(A) v_{m+1},
\end{eqnarray*}
where
\begin{equation} \label{gdef:eq}
g_m(t) =  \sum_{i=1}^p \frac{\omega_i \rho_m^{(i)}}{t-\sigma_i}.
\end{equation}
In the case of a rational function we can thus express the square of the 2-norm of the error as
\[
   \left(g_m(A)v_{m+1}\right)^H\left( g_m(A)v_{m+1} \right) = v_{m+1}^Hh_m(A)v_{m+1},
\]
where
\[
h_m(t) = \bar{g}_m(t)\cdot g_m(t)  = |g_m(t)|^2.
\]

We have thus shown that the square of the 2-norm and of the $A$-norm of the error of the CG iterate \eqnref{lanczos_cg:eq}
as well as of the Lanczos approximations \eqnref{lanczos_rat:eq} (only the 2-norm) are given in the form
\[
v^Hh(A)v, 
\]
where $h$ is a known rational function defined on $[\lambda_{\rm min},\lambda_{\rm max}]$ and $v$ is the current 
(the $m+1$-st) Lanczos vector. 

\section{Error bounds and error estimates} \label{bounds:sec} 
In this section we summarize the key aspects of the theory relating moments, quadrature and orthogonal polynomials, 
see \cite{golub-meurant93,golub-meurant97,golub-meurant-book}, which allows us to obtain
estimates and often even lower and upper bounds for quantities of the form $v^Hh(A)v$ and 
thus, in light of the discussion at the end of section~\ref{lanczos:sec}, for the norm of the error of the 
Lanczos approximations \eqnref{lanczos_cg:eq} and \eqnref{lanczos_rat:eq}. The error estimates obtained rely on running  
a new Lanczos process, now starting with $v$. 

Let $(\lambda_i,z_i), i=1,\ldots,n$ denote the eigenpairs of $A$ where the vectors $z_i$ are orthonormal and $\lambda_1 \leq \lambda_2 \leq \ldots \leq \lambda_n$. Expanding
$v$ in terms of the basis $z_i$ we can write
\[
v = \sum_{i=1}^n \gamma_i z_i. 
\]
Since $h(A)v = \sum_{i=1}^n h(\lambda_i) \gamma_i z_i$, cf.\ \cite{Frommer.Simoncini.06,Higham.book08}, we have
\begin{equation} \label{int:eq}
   v^Hh(A)v = \sum_{i=1}^n h(\lambda_i) \cdot |\gamma_i|^2 = \int_{a}^{b} h(t) \ d\gamma(t),
\end{equation}
where $[a,b] \supseteq [\lambda_{\min},\lambda_{\max}]$, the integral is to be understood as a Riemann-Stieltjes integral and the discrete measure $\gamma(t)$ is given as
\[
\gamma(t) = \left\{ \begin{array}{ll}
                  0 & \mbox{if }  t < \lambda_{\rm min} \\
                 \sum_{j=1}^i |\gamma_j|^2 & \mbox{if } \lambda_{i} \leq t < \lambda_{i+1} \\
                 \sum_{j=1}^n |\gamma_j|^2 & \mbox{if } \lambda_{n} \leq t   
                \end{array}
         \right. \ .
\]

We can now use Gauss, Gauss-Lobatto or Gauss-Radau quadrature rules to approximate $\int_{\lambda_{\rm min}}^{\lambda_{\rm max}} h(t) d\gamma(t)$. Algorithmically, evaluating these rules turns out to be very intimately related to the Lanczos process
based on the starting vector $v$. The precise results are as follows, see \cite{golub-meurant93,golub-meurant97,golub-meurant-book}.

\begin{theorem} \label{golub-meurant:thm}
Let $\widetilde{T}_k$ denote the tridiagonal matrix in the Lanczos relation \eqnref{lanczos_rel:eq} arising 
after $k$ steps of the Lanczos process with starting vector $v, \| v \| = 1$. Assume that $h$ is at least 
$2k$ times continuously differentiable on an open set containing $[a,b]$.
\begin{itemize}
  \item[(i)] Approximating \eqnref{int:eq} with the Gauss quadrature rule using  $k$ nodes 
$t_j \in (a,b)$ gives
           \[
                  v^Hh(A)v = e_1^Hh(T^{\rm G}_k)e_1 + R^{\rm G}_k[h], \mbox{ where } T_k^{\rm G} = \widetilde{T}_k,
           \]
           with the error $R^{\rm G}_k[h]$ given as  
\begin{equation} \label{Gauss-err:eq}
   R^{\rm G}_k[h] = \frac{ h^{(2k)}(\xi)}{ (2k)! }
        \int_a^b 
          \left[ \prod_{j = 1}^k (t - t_j) \right]^2
          \ d \gamma (t),
          \quad a < \xi < b 
          \ .
\end{equation}
  \item[(ii)] Approximating \eqnref{int:eq} with the 
             Gauss-Radau quadrature rule using  $k-1$ nodes $t_j \in (a,b)$ 
             with one additional node fixed at $a$ gives 
             \[
                  v^Hh(A)v = e_1^Hh(T_k^{\rm GR})e_1 + R_k^{\rm GR}[h].
              \]
             Here, the tridiagonal matrix $T_k^{\rm GR}$ differs from $\widetilde{T}_k$ in that its $(k,k)$ entry $\alpha_k$ is 
             replaced by $\widetilde{\alpha}_k = a +\delta_{k-1}$, where $\delta_{k-1}$ is the last entry of the vector $\delta$ with $(\widetilde{T}_{k-1}-aI) \delta = \beta_ {k-1}^2 e_{k-1}$.
            The error $R^{\rm GR}_k[h]$ is given as  
\begin{equation} \label{Gauss-Radau-err:eq}
   R_k^{\rm GR}[h] = \frac{ h^{(2k-1)}(\xi)}{ (2k-1)! }
        \int_a^b 
          (t -a) \left[ \prod_{j = 1}^{k-1} (t - t_j) \right]^2
          \ d \gamma (t),
          \quad a < \xi < b 
          \ .
\end{equation}
\item[(iii)] Approximating \eqnref{int:eq} with the 
             Gauss-Lobatto quadrature rule using  $k-2$ nodes $t_j \in (a,b)$ and
             two additional nodes, one fixed at $a$ and one fixed at $b$, gives 
\[
                  v^Hh(A)v = e_1^Hh(T_k^{\rm GL})e_1 + R^{\rm GL}_k[h].
              \]
              Here, the tridiagonal matrix $T_k^{\rm GL}$ differs from $\widetilde{T}_k$ in its last column and row. With $\delta$ and $\mu$ the solutions of the system $(\widetilde{T}_{k-1}-aI)\delta = e_{k-1}$, $(\widetilde{T}_{k-1}-bI)\mu = e_{k-1}$ and $\widetilde{\alpha}_{k}, \widetilde{\beta}_{k-1}^2$ the solution of the linear system 
               \[
                 \left[ \begin{array}{cc} 1 & -\delta_k \\
                                          1 & -\mu_k
                        \end{array}
                  \right] \left[ \begin{array}{c} \widetilde{\alpha}_k \\ \widetilde{\beta}_{k-1}^2 \end{array} \right] = 
                   \left[ \begin{array}{c} a \\ b \end{array} \right] ,
               \]
               the tridiagonal matrix $T_k^{\rm GL}$ is obtained from $\widetilde{T}_k$ by replacing $\alpha_k$ by 
               $\widetilde{\alpha}_k$ and  $\beta_{k-1}$ by $\widetilde{\beta}_{k-1}$. The error $R^{\rm GL}_k[h]$ is given as  
\begin{equation} \label{Gauss-Lobatto-err:eq}
   R_k^{\rm GL}[h] = \frac{ h^{(2k-2)}(\xi)}{ (2k-2)! }
        \int_a^b 
          (t -a)(t-b) \left[ \prod_{j = 1}^{k-2} (t - t_j) \right]^2
          \ d \gamma (t),
          \quad a < \xi < b 
          \ .
\end{equation}
\end{itemize}
\end{theorem}

Inspecting the quadrature error terms $R^{\rm G}_k[h]$, $R^{\rm GR}_k[h]$ and $R^{\rm GL}_k[h]$, we get the following corollary which applies Theorem~\ref{golub-meurant:thm} to the rational functions $h$ through which we expressed the error of the $m$-th Lanczos approximation 
as $v_{m+1}^Hh(A)v_{m+1}$ at the end of section~\ref{lanczos:sec}. The corollary is thus the key to obtaining error bounds for the Lanczos approximations.   

\begin{corollary} \label{bounds:cor}
The estimates $e_1^Hh(T_k^{\rm G})e_1$, $e_1^Hh(T_k^{\rm GR})e_1$ and $e_1^Hh(T_k^{\rm GL})e_1$ 
from Theorem~\ref{golub-meurant:thm} (i), (ii) and (iii), resp., represent lower or upper bounds for \eqnref{int:eq} if the derivatives $h^{(2k)}, h^{(2k-1)}$ 
and $h^{(2k-2)}$ have constant sign on the interval $[a,b]$. 

This is true in particular for the rational functions $h(t) =  \rho_m^2 t^{-1}, h(t) = \rho_m^2t^{-2}$ from \eqnref{linsys_error:eq}
as well as $h(t) = g^2_m(t)$ with 
\[
g_m(t) = \sum_{i=1}^p \frac{\omega_i\rho_m^{(i)}}{t-\sigma_i} \enspace \mbox{with } \omega_i \geq 0,
\sigma_i \leq 0, i=1,\ldots,p,
\]
for which $h^{(2k)}(t) \geq 0, h^{(2k-1)} \leq 0$ for $t \in (0,\infty)$ 
and $k \in \mathbb{N}.$
\end{corollary}
\begin{proof} The only non-trivial part of the corollary concerns the derivatives of $h(t) = 
g_m^2(t)$. We first note that by Lemma~\ref{cgres_lanczosvec:lem}, $\sign(\rho_m^{(i)}) = (-1)^{m}$, independently of $i$. Thus, the derivatives of each of the summands of $g_m(t)$ 
have constant sign on $[0,\infty)$, resulting in 
\[
   \sign\left(\frac{d^{\ell} g_m(t)}{dt^{\ell}} \right) = (-1)^{\ell+m} \mbox{ for all } t \in [0,\infty).
\] 
Using
\[
\frac{d^{\ell}h_m(t)}{dt^{\ell}} = \sum_{j=0}^{\ell} \left( \begin{array}{c} \ell \\ j \end{array} \right) 
\frac{d^{j}g_m(t)}{dt^j} \cdot \frac{d^{\ell-j}g_m(t)}{dt^{\ell-j}} 
\]
we thus see that $d^{\ell} h_m(t)/dt^{\ell} < 0 \; ( > 0)$ for $t \in [0,\infty)$ if $\ell$ is odd  (even). 
\end{proof}

We just note that there is a connection to results from \cite{Druskin08,fro08} on the monotone convergence of the Lanczos approximations. 

For future reference we state the computational cost of the error estimates from Theorem~\ref{golub-meurant:thm} for those functions $h$ of interest in this paper.

\begin{lemma} \label{cost:lem}
Assume that $\widetilde{T}_k$ is given. Let $h(t) = t^{-1}$ or $h(t) = t^{-2}$ or $h(t) = \bar{g}(t)g(t)$ with $g(t) = \sum_{i=1}^p \frac{\omega_i}{t-\sigma_i}$. Then the cost for evaluating the estimates from Theorem~\ref{golub-meurant:thm} (i), (ii) and (iii) 
is $\mathcal{O}(k)$.
\end{lemma}
\begin{proof}
Solving a linear system with a tridiagonal matrix of size $k$ has cost $\mathcal{O}(k)$. Thus the cost for obtaining the matrices $T_k^{\rm GR}$ and $T_k^{\rm GL}$ from parts (ii) and (iii) is $\mathcal{O}(k)$. Denote by $T_k$ any of the matrices $T_k^{\rm G}, T_k^{\rm GR}$ and $T_k^{\rm GL}$. For the case $h(t) = |\rho_m|^2\cdot t^{-1}$ we have to solve the linear system $T_k y = e_1$ and to compute $e_1^Hy$ which has cost $\mathcal{O}(k)$. Similarly, 
for $h(t)= |\rho_m|^2\cdot t^{-2}$ we have to solve $T_k y = e_1$ and compute $y^Hy$, which has again cost $\mathcal{O}(k)$. Finally, for $h_m(t) = \bar{g}_m(t)g_m(t)$ we have to solve $(T_k - \sigma_i I)y^{(i)} = e_1$ for $i=1,\ldots,p$, compute $y = \sum_{i=1}^p \omega_i y^{(i)}$ and then $y^Hy$, which has total cost $\mathcal{O}(pk)$ which is $\mathcal{O}(k)$ if we consider $p$ as fixed. 
\end{proof}

It is important to note that if one considers evaluating the error estimates for a {\em sequence} of values for $k$, most of the quantities needed can be obtained by an update from $k$ to $k+1$ with cost $\mathcal{O}(1)$ only. For details we refer to \cite{golub-meurant-book}, e.g.

\section{Lanczos restart recovery}
We want to use the results of Theorem~\ref{golub-meurant:thm} to obtain bounds
or estimates for the error of the iterate $x_m$ of the CG iterate
\eqnref{lanczos_cg:eq} or the Lanczos approximation for a rational function
\eqnref{lanczos_rat:eq}. To avoid ambiguities, let us call the Lanczos process 
via which the iterates $x_m$ are obtained the {\em primary}
Lanczos process.
The straightforward way to obtain the error estimates from
Theorem~\ref{golub-meurant:thm} would be to perform $k$ steps of a new, {\em
restarted} Lanczos process which takes the current Lanczos vector
$v_{m+1}$ of the primary process as its starting vector. 
This results in the restarted Lanczos relation
\begin{equation} \label{eq:restartedlanczosrelation}
AV^{\rr}_{k} = V_{k+1}^{\rr}\overline{T}_{k}^{\rr}, \mbox{ where } V_k^{\rr} = \left[v_1^{\rr} \mid \ldots \mid v_k^{\rr}\right], V_{k+1}^{\rr} = \left[ V_k^{\rr} \mid v_{k+1}^{\rr}\right], v_1^{\rr} = v_{m+1},
\end{equation} 
and we can now apply the theorem using the tridiagonal matrix
$T^{\rr}_k$ arising from the restarted process. This is,
however, far too costly in practice: computing the error estimate would
require $k$ multiplications with $A$---approximately the same amount of work
that we would need to advance the primary iteration from step $m$ to $m+k$.    

Fortunately, as we will show now, it is possible to cheaply retrieve the
matrix $T^{\rr}_k$ of the secondary Lanczos process from the matrix
$T_{m+k+1}$ of the primary Lanczos process. This {\em Lanczos restart recovery}
opens the way to efficiently obtain all the error estimates from
Theorem~\ref{golub-meurant:thm} in a retrospective manner: At iteration $m+k$
we get the estimates for the error at iteration $m$ without using any
matrix-vector multiplications with $A$ and with cost $\mathcal{O}(k^2)$, independently
of the system size $n$. 

For $m= 0,1,\ldots$ fixed, we define the tridiagonal matrix $\hat{T}_{2k+1}$ as 
the block of $T_{m+1+k}$ ranging from rows and columns $\max\{1,(m+1)-k\}$ to $(m+1)+k$. This means that $\hat{T}_{2k+1}$ is the trailing $(2k+1) \times (2k+1)$ diagonal sub-matrix of $T_{m+k+1}$, except for $m+1 \leq k$, where its size is
$(m+1)+k \times (m+1)+k$.

The following theorem shows that for Lanczos restart recovery we basically have to run the 
Lanczos process for the tridiagonal matrix $\hat{T}_{2k+1}$, starting with the $k+1$-st unit vector $e_{k+1} \in \mathC^{2k+1}$.

\begin{theorem} \label{lanczos_recovery:thm}
Let the Lanczos relation for $k$ steps of the Lanczos process for $\hat{T}_{2k+1}$ with starting vector 
$e_{k+1} \in \mathC^{2k+1}$ ($e_{m+1} \in \mathC^{m+1+k}$ if $m+1 \leq k$) be given
as
\begin{equation} \label{reduced-lanczos:eq}
  \hat{T}_{2k+1} \hat{V}_k = \hat{V}_{k+1} \widetilde{\overline{T}}_{k}.
\end{equation}
Then the matrix $\widetilde{\overline{T}}_{k}$ is identical to $\overline{T}_{k}^{\rr}$ from the restarted Lanczos relation \eqnref{eq:restartedlanczosrelation},
\begin{equation} \label{recovery:eq}
\overline{T}^{\rr}_{k} = \widetilde{\overline{T}}_{k}.
\end{equation} 
\end{theorem}

\begin{proof}
For notational simplicity, we only consider the case $m+1 >k$ where $\hat{T}_{2k+1}$ 
has its full size $(2k+1) \times (2k+1)$. 
Recall the Lanczos relation for $m+k$ steps of the primary Lanczos process given in  \eqnref{lanczos_rel:eq}, 
\[
  AV_{m+k} = V_{m+k+1}\overline{T}_{m+k}.
\]

Since $v_{m+1} \in K_m(A,v_1)$ we have  $K_{k+1}(A,v_{m+1}) \subseteq K_{m+k+1}(A,v_1)$. Hence we can express the vectors $v_i^{\rr}, i=1,\ldots,k+1$ of the restarted 
Lanczos process, see \eqnref{eq:restartedlanczosrelation}, in terms of a basis of $K_{m+k+1}(A,v_1)$. The columns of $V_{m+k+1}$ form such a basis, i.e., we have
\[
v_i^{\rr} = V_{m+k+1} q_i, \, q_i \in \mathC^{m+k+1}, i=1,\ldots,k+1.
\]
The vectors $q_i$ are orthonormal, since the vectors $v_i^{\rr}$ and the columns of $V_{m+k+1}$ are orthonormal, too. Putting $Q_i = \left[ q_1 \mid \ldots \mid q_i \right] \in \mathC^{(m+k+1) \times i}$ we thus have
\[
 V_i^{\rr} = V_{m+k+1} Q_i, \, i=1,\ldots,k,
\]
so that the restarted Lanczos relation \eqnref{eq:restartedlanczosrelation} can be written as
\begin{equation} \label{intermediate_lanczos:eq} 
AV_{m+k+1}Q_k = V_{m+k+1}Q_{k+1} \overline{T}^{\rr}_{k}.
\end{equation}
All columns of $V_{m+k+1}Q_k = V_{k}^{\rr}$ are from $K_k(A,v_{m+1}) \subseteq K_{m+k}(A,v_1)$,
so the columns of $AV_{m+k+1}Q_k$ are all from $K_{m+k+1}(A,v_1)$, on which the projector $V_{m+k+1}V^H_{m+k+1}$ acts as the identity. From \eqnref{intermediate_lanczos:eq} we therefore get
\[
V_{m+k+1}\underbrace{V^H_{m+k+1}AV_{m+k+1}}_{= T_{m+k+1} \enspace {\rm by~\eqnref{eq:restartedlanczosrelation}}}Q_k = V_{m+k+1}Q_{k+1} \overline{T}^{\rr}_{k},
\]
and since $V_{m+k+1}$ has full column rank we have
\begin{equation} \label{intermediate_lanczos2:eq}
T_{m+k+1}Q_k = Q_{k+1} \overline{T}^{\rr}_{k}.
\end{equation}

\begin{figure}
\begin{center}\input{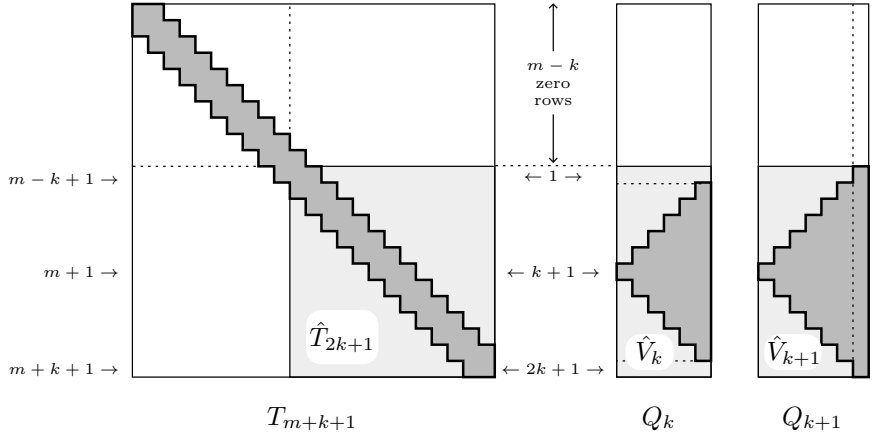}
\end{center}
\caption{Illustration for the proof of Theorem~\ref{lanczos_recovery:thm}. Dark grey: non-zero entries; light grey and indices in the middle: restarted Lanczos; indices on the left: primary Lanczos. \label{illustration:fig}} 
\end{figure}

The matrix $Q_{k+1}$ has a special sparsity pattern:
Since $v_1^{\rr} = v_{m+1}$, we have $q_1 = e_{m+1}$, the $m+1$-st unit vector in
$\mathC^{m+k+1}$.  The matrix $T_{m+k+1}$ being tridiagonal, a trivial induction shows that $q_i$ holds non-zeros only in those components $j$ for which $|j-(m+1)| < i$, see also Figure~\ref{illustration:fig}.
Consequently, $Q_k$ has non-zeros only in rows $m-k+2$ to $m+k$, and $Q_{k+1}$ only in rows $m-k+1$ 
to $m+k+1$. Defining $\hat{V}_{k}$ and $\hat{V}_{k+1}$ as the matrices consisting of the last $2k+1$ rows
(rows $m-k+1$ to $m+k+1$) of $Q_k$ and $Q_{k+1}$, respectively, we see that $\hat{V}_{k+1}$ is identical to  $[ \hat{V}_k \mid \hat{v}_{k+1} ]$, that it has orthonormal columns and that $\hat{v}_1 = e_{k+1}$, the $k+1$-st unit vector in $\mathC^{2k+1}$. Moreover, with $\hat{T}_{2k+1}$ as defined in the theorem, we obtain from \eqnref{intermediate_lanczos2:eq} 
\[
\hat{T}_{2k+1}\hat{V}_{k} = \hat{V}_{k+1} \overline{T}^{\rr}_k.
\]
Due to the essential uniqueness of the Lanczos relation, Lemma~\ref{invariance:lem}(ii),
this finishes our proof. \end{proof} 

The above theorem shows that we can retrieve ${T}_{k}^{\rr}$ from ${T}_{m+k+1}$ by performing $k$ steps of
the Lanczos process for the $(2k+1)\times(2k+1)$ tridiagonal matrix $\hat{T}_{2k+1}$. Here each step needs
$\mathcal{O}(k)$ operations\footnote{Since $\tilde{v}_j$ is non-zero only in positions $k+1-(j-1),\ldots,k+1+(j-1)$, step $j$ actually has only cost $\mathcal{O}(j)$. This refined analysis does, however, not affect the $\mathcal{O}$-analysis of the total cost}, so that the overall cost for computing ${T}_{k}^{\rr}$ is $\mathcal{O}(k^2)$. Together with 
Lemma~\ref{cost:lem} we conclude that the total cost for computing the error estimates from Theorem~\ref{golub-meurant:thm} 
is also $\mathcal{O}(k^2)$ . 

Algorithm~\ref{LanczosPlusEstimates:alg} shows how we suggest to use the results exposed so far. It computes the Lanczos
approximations $x_m$ for $g(A)b$ with $g(t) = \sum_{i=1}^p
\frac{\omega_i}{t-\sigma_i}$ and the estimates $\ell_{m-k}, u_{m-k}$ for the error at
iteration $m-k$ based on the Gauss and the Gauss-Radau rule. By Corollary~\ref{bounds:cor},
these estimates represent lower and upper bounds, respectively, if all poles $\sigma_i$ are negative and $\omega_i \geq 0$ for all $i$. The algorithm can be modified to also obtain error estimates or bounds based on the Gauss-Lobatto rule and to get bounds for the $A$-norm in case we deal with a linear system.

\begin{algorithm}
  \caption{Lanczos approximations for rational function with 2-norm error estimates/bounds }
  \label{LanczosPlusEstimates:alg}
  set $x_{-1} = 0$, $\rho_0 = \| b \| $, $ \tau_0 = 1$ \;
  choose $k$ \;
  \For{ $m = 0, 1, \dots$ }{

    compute $\alpha_{m+1}$, $\beta_{m+1}$, $v_{m+2}$ using the Lanczos process for $A$ \;
    \For(\tcc*[f]{loop over poles}){ $i=1,\ldots,p$}{
    \If{ $m > 0 $ }{
      $\tau_m^{(i)} = \left[ 
        1 - 
        \frac{ \alpha_m -\sigma_i }{ \alpha_{m+1} -\sigma_i}
        \left( \frac{ \rho_m^{(i)} }{ \rho_{m-1}^{(i)} } \right)^2
        \frac{ 1 }{ \tau_{m-1}^{(i)} }
      \right]^{-1} $
    }
    $ \rho_{m+1}^{(i)} 
    = - \tau_m^{(i)} \rho_m^{(i)} \tfrac{ \beta_{m+1} }{ \alpha_{m+1} -\sigma_i}$\;
    $ x_{m+1}^{(i)} = \tau_m^{(i)} \left(x_m^{(i)} + \tfrac{\rho_m^{(i)}}{\alpha_{m+1}-\sigma_i} v_{m+1}\right) + \left(1 - \tau_j^{(i)}\right) x_{m-1}^{(i)} $ \;
   }
   $x_{m+1} = \sum_{i=1}^p \omega_i x_{m+1}^{(i)}$ \;
   \If{$m > k$} {
      perform $k$ steps of the Lanczos process for the trailing $(2k+1) \times (2k+1)$ diagonal sub-matrix of $T_{m+1}$, 
      this yields the tridiagonal matrix $T_k^{\rr} \in \mathC^{k \times k}$ \;
      $\ell_{m-k} = \| g_m({T}^{\rr}_k)e_1\|_2  $ \tcc*{$g_m$ is given in \eqnref{gdef:eq}} 
      $u_{m-k} =  \| g_m({T}^{\rm GR}_k)e_1\|_2 $ \tcc*{$\widetilde{T}^{\rm GR}_k$ given in Theorem~\ref{golub-meurant:thm}(ii)}  
   }
  }
\end{algorithm}

\section{Comparison with existing methods} \label{comparison:sec}
Let us first consider a single linear system $Ax=b$. If we solve this system via the CG method, we (implicitly) perform a ``primary'' Lanczos process. Assume that $n$ iterations give the exact solution $x_n = \|b\| \cdot V_n T_n^{-1}e_1$, see \eqnref{lanczos_cg:eq}. Then the $A$-norm of the  error of the $m$-th iterate, $\|x_m-x_n\|_A$, can be expressed as
\[ \label{cg_error_A:eq}
   \|x_m-x_n\|_A^2 = \| b \|^2 \cdot \left( e_1^HT_n^{-1}e_1 - e_1^H T_m^{-1}e_1 \right).
\]
The matrix $T_m$ is available, in principle, from the CG iteration. Its Cholesky factorization
can be updated easily from one step to the next. Also, it can be shown that $e_1^H T_m^{-1}e_1$ is a positive number which increases montonically with $m$. This implies that $\eta_{k,m} := \|b\|^2\cdot (e_1^HT_{m+k}^{-1}e_1 - e_1^HT_m^{-1}e_1)$ is a lower bound for $\|x_m-x_n\|_A^2$ for any $k > 0$. The challenge is to obtain a numerically stable way to update $\eta_{k,m}$ as the CG iteration proceeds. Starting with \cite{daesgo72,dagona78}, many papers have been devoted to this topic, see \cite{golub-meurant97,Golub.Strakos.94,Meurant.97,Meurant.99,Meurant.05,Strakos.Tichy.02,Strakos.Tichy.05}, summarized in Golub's and Meurant's book \cite{golub-meurant-book}. In order to also obtain {\em upper} bounds for the $A$-norm of the error---provided bounds on the spectrum of $A$ are known---the approach sketched so far can be extended to include Gauss-Radau and Gauss-Lobatto type estimates by (implicitly) using the matrices $T_k^{\rm GR}$ and $T_k^{\rm GL}$ defined in 
Theorem~\ref{golub-meurant:thm}. Meurant's CGQL algorithm (CG with Lanczos quadrature) 
from \cite{Meurant.97} (see also \cite{golub-meurant-book} and \cite{meurant-tichy11}) does so and thus computes upper and lower bounds for the $A$-norm of the error along with the CG iterates.  
As was already noted in \cite{hest-stiefel}, $\eta_{k,m}$ can be computed using just the CG coefficients of iterations $m$ to $m+k$. This is based on the fact that the CG-algorithm in its standard form updates the iterates as $x_{m+1} = x_m + \gamma_m p_m$,
where the search directions are $A$-orthogonal so that $ \| x_{m+k} - x_m\|_A^2 = \sum_{i=0}^{k-1} 
|\gamma_{m+i}|^2 \cdot \|p_{m+i}\|_A^2$. The CGLQ algorithm \cite{Meurant.97} implements this approach, and the analysis by
Strako\v{s} and Tich\'y \cite{Strakos.Tichy.02,Strakos.Tichy.05} shows that this approach represents to date the most stable computation of $\eta_{k,m}$, because problems due to loss of orthogonality in the primary Lanczos process are eliminated. Recently,  the paper \cite{meurant-tichy11} shows how to transport this approach to also obtain upper bounds on the $A$-norm of the error. In all these approaches, the additional cost for getting the estimates and bounds is $\mathcal{O}(k)$ per iteration. We note that the approach presented in the present work also respects the philosophy, motivated in \cite{Strakos.Tichy.05}, e.g., to rely on local orthogonality only. Indeed, we just work with the quantities from iterations $m-k,\ldots,m+k$ when computing the error estimate for iteration $m$.  
 
For the 2-norm of the error one can use the relation, cf.\ \cite[Corollary 21.7]{golub-meurant-book},
\[
\|x_m - x_n \|^2 = \| b \|^2 \cdot \left( e_1^HT_n^{-2}e_1 - e_1^HT_m^{-2}e_1\right) - 2 \frac{e_m^HT_m^{-2}e_1}{e_m^HT_m^{-1}e_1} \| x_m-x_n\|^2_A.
\]
and proceed in a similar manner as before. This is explained in detail in \cite{golub-meurant-book}; the resulting estimate is not necessarily a lower bound any more. Another estimate for the $2$-norm of the error was proposed in \cite{Strakos.Tichy.02}. It involves the CG coefficients of iterations $m$ to $m+2k$ to obtain an estimate for the error at iteration $m$ and it was shown there that this estimate is actually a lower bound. Note that none of the approaches 
for the 2-norm estimates has a ``Gauss-Radau'' counterpart which would allow for estimates that represent upper bounds. 

For the case of a rational matrix function $g(A)b$, with iterates obtained via the Lanczos approximation \eqnref{lanczos_rat:eq}, the paper \cite{FrSi07} extends the approach of \cite{Strakos.Tichy.02} to get estimates for the 2-norm of the error. If all poles $\sigma_i$ are negative and all coefficients $\omega_i$ are positive, these estimates were proven to be lower bounds in \cite{FrSi07}. If there are complex poles, the estimates in \cite{FrSi07} were derived 
for the CG method based on the bilinear form $\langle x,y\rangle = y^Tx$ rather than $y^Hx$. Again, there is no variant which would compute upper bounds for the error in the 2-norm. Therefore, in \cite{FrSi09}, a different approach was used to get an upper bound: Using a global optimization algorithm, the 
maximum $c_m = \max_{t \in [\lambda{\min},\lambda_{\max}]}|g_m(t)|$ for $g$ from \eqnref{gdef:eq} is bounded from above by $\bar{c}_m$ which then is a bound for the error in the 2-norm. While one succeeds in getting upper bounds for the error with this approach, it is quite costly due to the global minimization, and the upper bounds are not necessarily very precise. 

As an alternative to the quadrature approach for the case of the CG iteration, the paper \cite{brezinski99} suggests to use vector extrapolation on the current residual $r$ and $Ar$ to obtain error estimates. They can be modified to yield bounds, 
provided the smallest and largest eigenvalue of $A$ are known. The computation of $Ar$ can be avoided by using quantities available from the Lanczos process, but the extrapolation requires 
inner products with full vectors. Hence the cost for the error estimates of \cite{brezinski99} is $\mathcal{O}(n)$.

\section{Numerical experiments}
In this section we illustrate the quality of the error bounds developed in this paper for several rational functions arising from applications. All experiments were carried out on a standard workstation using Matlab R2011a.

\begin{figure} 
\centerline{\includegraphics[width=0.48\textwidth]{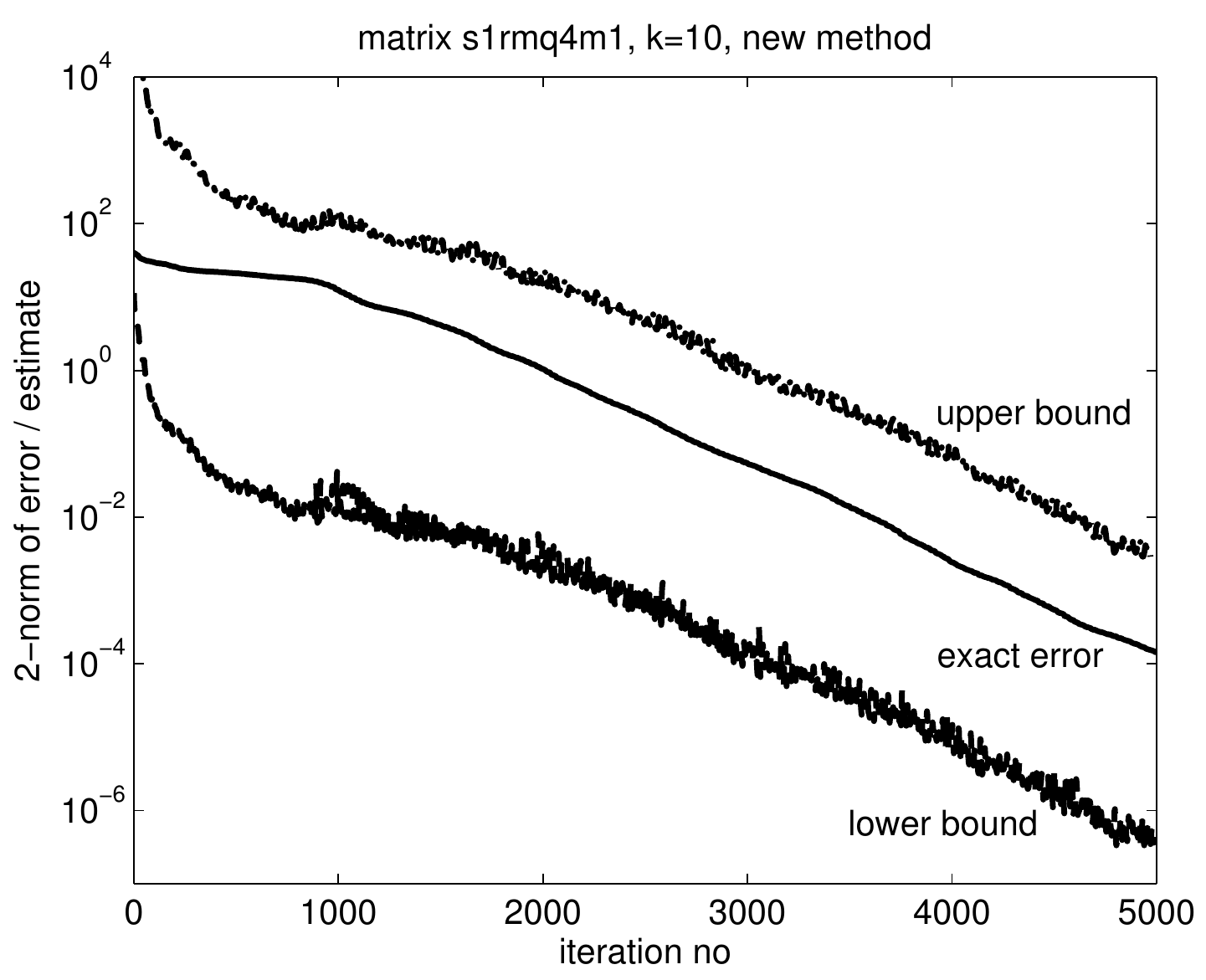}
\hfill \includegraphics[width=0.48\textwidth]{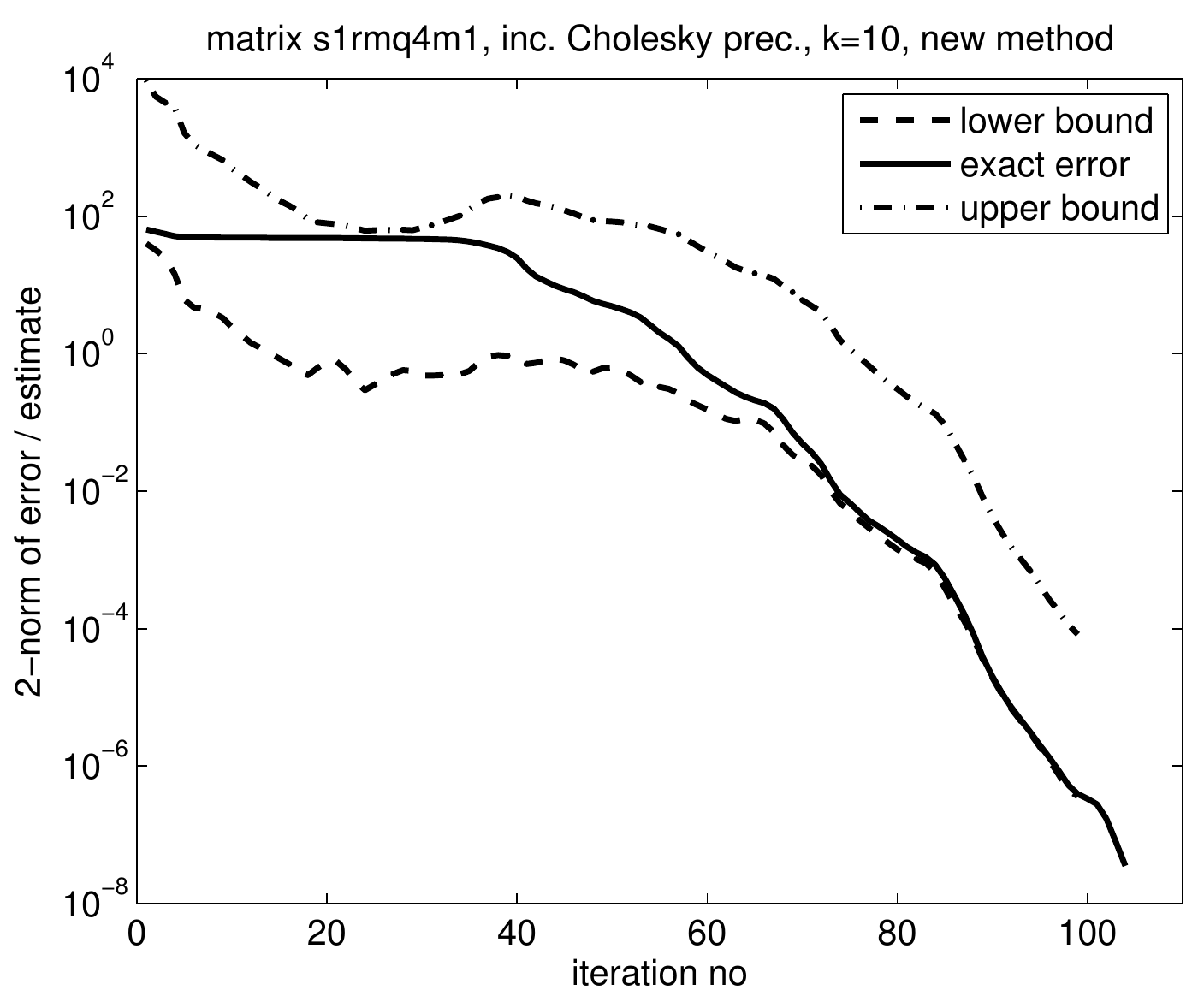}}
\centerline{\includegraphics[width=0.48\textwidth]{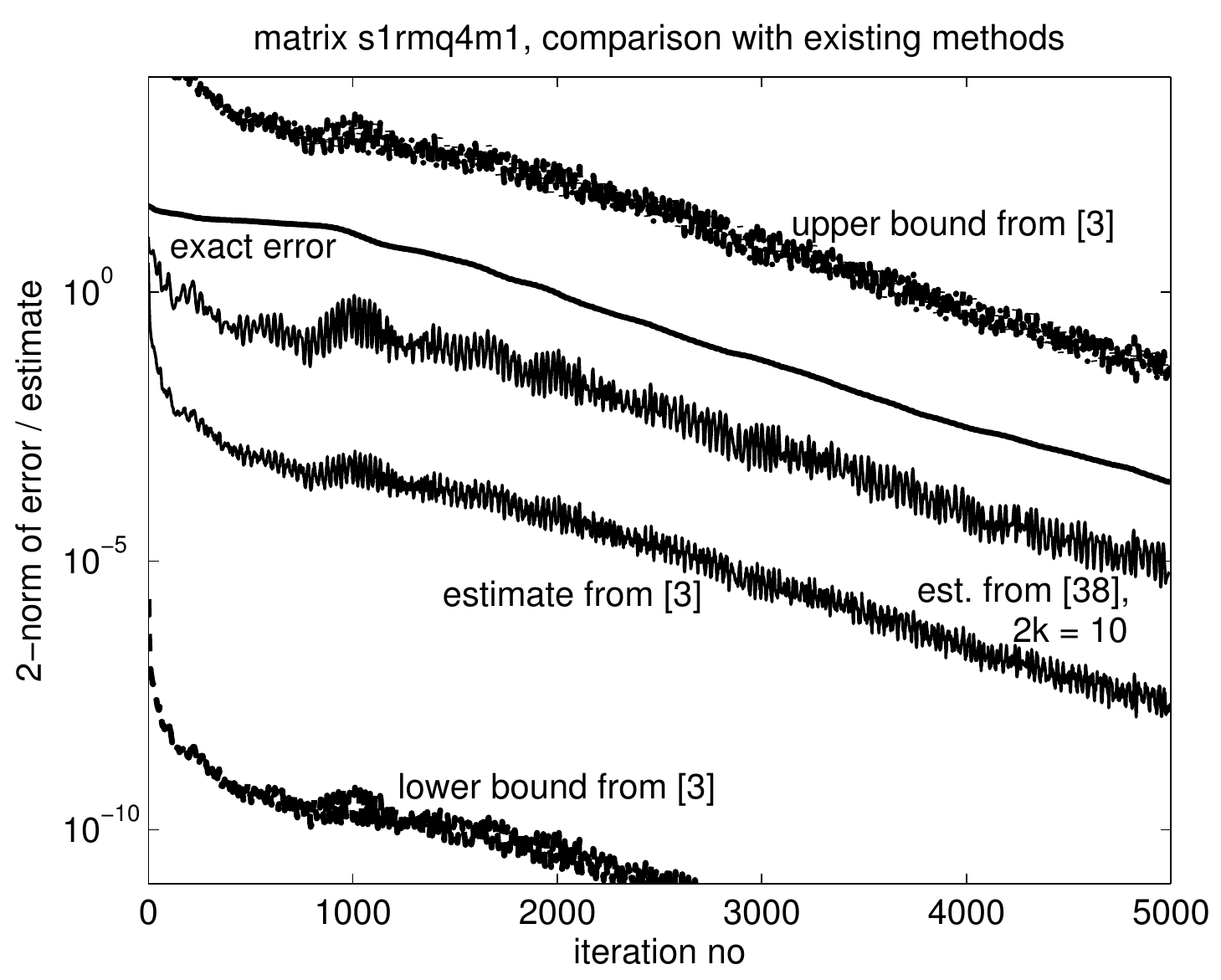}
\hfill \includegraphics[width=0.48\textwidth]{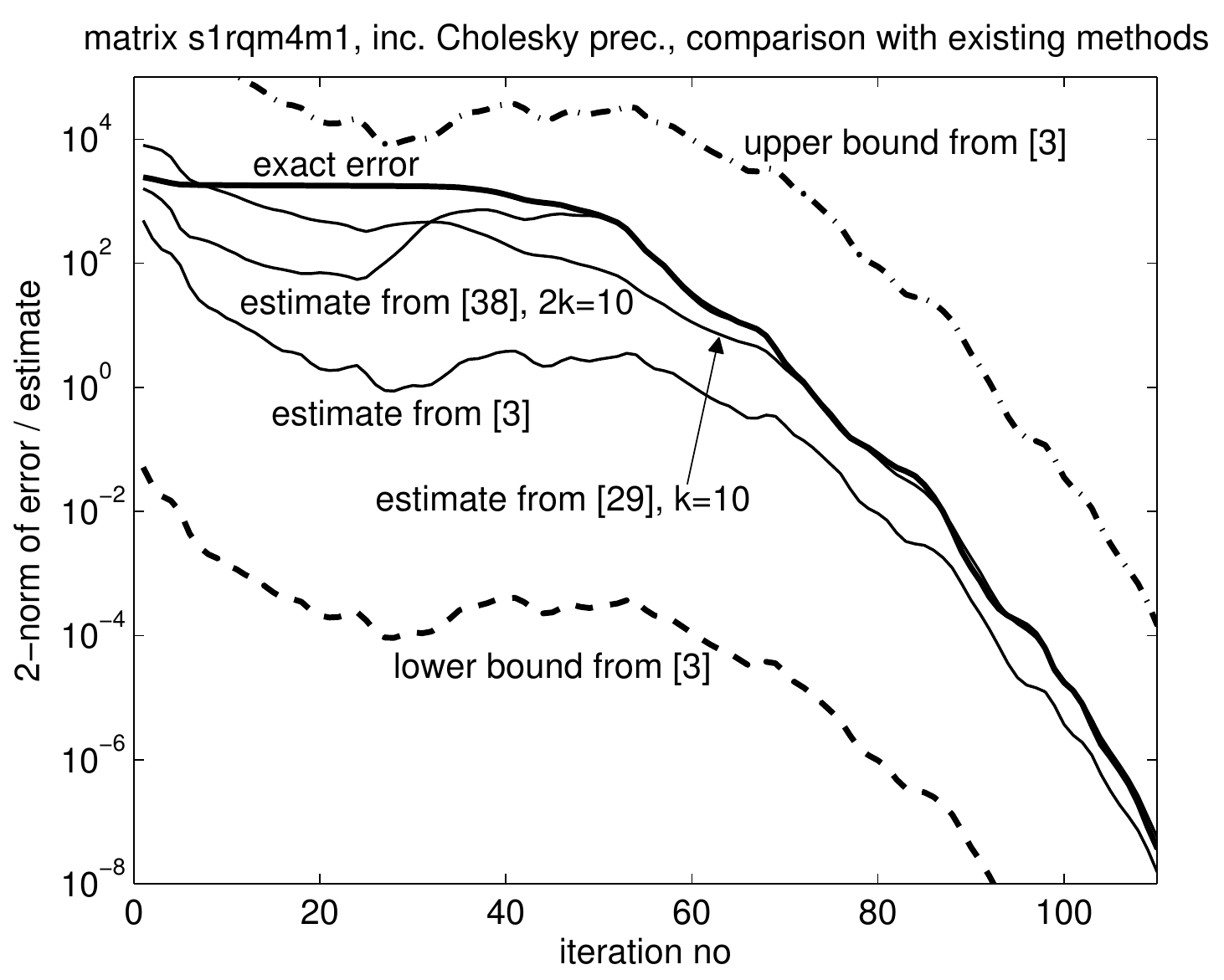}}
\caption{Error bounds and exact error for the CG iterates for $Ax=b$, $A$ matrix {\tt s1rmq4m1} for $k=10$. Left column:
no preconditioning,  right column: incomplete Cholesky preconditioning. Top row: Algorithm~\ref{LanczosPlusEstimates:alg}, bottom row: comparison with other methods.\label{s1:fig}}
\end{figure}

As a first example, we look at the CG iterates for a linear system $Ax = b$, where $A$ is the
matrix {\tt s1rmq4m1} from the matrix group {\tt Cylshell} of the University of Florida matrix collection, see \cite{UFL,Davis-Hu11}. It arises from a finite element discretization of a 
cylindrical shell. We chose this example because of its relatively high condition number which is  
of the order of $10^6$, despite its small size ($n=5,489$).  The solution $x$ was generated randomly, then explicitly computing $b= Ax$ as the right hand side. We do not {\em a priori} know a lower bound $a$ on the spectrum. We thus monitored the smallest eigenvalue of the tridagonal matrix $T_m$ which is known to converge to $\lambda_{\min}$ from above. 
As soon as for some iteration, $k_0$ say, the relative change in the smallest eigenvalue $\underline{\lambda}$ was less than $10^{-4}$, we
put $a = 0.99 \cdot \underline{\lambda}.$ Hence, in principle, we obtain upper bounds only for iterations beyond $k_0$. To have the complete picture, though, we afterwards added the upper bounds obtained with this value for $a$ for iterations 1 to $k_0-1$.  
For the unpreconditioned system (left column of Figure~\ref{s1:fig}) we see that the CG method converges very slowly. Even with $k=10$, the error bounds obtained with Algorithm~\ref{LanczosPlusEstimates:alg} are quite severe over- and underestimations of the exact error. The second row gives error estimates
obtained with other methods. On the one hand, we used the method from \cite{Strakos.Tichy.02} in which we chose the parameter $k=5$, meaning that we use information from the next $2k = 10$  CG iterations, just as we do in Algorithm~\ref{LanczosPlusEstimates:alg} with $k=10$.
On the other hand, we also tested the recommended, third extrapolation based method from \cite{brezinski99}. This method can be modified to also give error bounds, provided we know the extremal eigenvalues of $A$, and the results for these bounds are also given.
We can see that the bounds and estimates obtained with the extrapolation based methods are worse than the bounds obtained via the new approach. The same holds for the other extrapolation based methods from \cite{brezinski99}, for which we do not reproduce the results here.  The estimate from \cite{Strakos.Tichy.02}, which actually is a lower bound in this case, is comparable to and slightly more accurate than the lower bound obtained with Algorithm~\ref{LanczosPlusEstimates:alg}.     
The right column of Figure~\ref{s1:fig} shows that we get much faster convergence and 
much better error bounds  when we use a standard, 
0-fill incomplete Cholesky factorization of $A$ as a preconditioner. This means that we perform Algorithm~\ref{LanczosPlusEstimates:alg} for the matrix $L^{-1}AL^{-H}$ instead of $A$ and $L^{-1}b$ instead of $b$,
where $A = LL^H+R$ is the incomplete Cholesky factorization of $A$. The bound $a$ for the smallest eigenvalue of the preconditioned matrix was obtained as before. A comparison between the new approach (first row) and existing methods (second row) leads to similar conclusions as in the non-preconditioned case. We also included results obtained with the method from \cite{Meurant.05} for comparison. The error estimates obtained for comparable parameters ($k=10$) appear quite similar to 
those obtained with our new method. For the non-preconditioned matrix
the error estimates with the method from \cite{Meurant.05} were extremely small (e.g., $10^{-30}$ after 1000 iterations, and even less for later iterations), so that we did not reproduce them in the left column of Figure~\ref{s1:fig}. We suspect that in this case the bad conditioning induces an instability which we could not avoid although we used the most stable, $QR$-based implementation of the method from \cite{Meurant.05}.

  It can be noticed that for the preconditioned system the upper bounds obtained with Algorithm~\ref{LanczosPlusEstimates:alg} are not as close to the exact error than the lower bounds, whereas it is the other way around for the non-preconditioned system. As a rule, we observed that for better conditioned systems, the lower bounds tend to be closer than the upper bounds, even if we work with very accurate estimates $a$ for the smallest eigenvalue. A theoretical justification for this behavior is still missing. 

Our second example deals with rational approximations to the sign function, as
it is used within the Neuberger overlap operator in lattice QCD. QCD (quantum chromodynamics) 
is the physical theory of quarks and gluons as the constituents of matter. To evaluate this
theory non-perturbatively, one has to work with discretizations on a 4-dimensional space-time lattice amongst which the Wilson fermion  matrix $I - \kappa D_W$ is the most important one. $D_W$ describes a nearest neighbor coupling on an equispaced 4d grid where each grid point holds 12 variables. Recent progress aiming at preserving the physically important
``chiral symmetry'' (cf.\ \cite{Boetal05}) on the lattice lead to Neuberger's overlap operator $D_N$ which has the form $P+\sign(PD_W)$, $P$ a simple
unitary matrix. The matrix $PD_W$ is hermitian and indefinite. In order to solve systems 
with $D_N$ one uses a Krylov subspace method, so that in each step one has to compute
$\sign(PD_W)b$ for some vector $b$. The matrix $Q:= PD_W$ is hermitian and
indefinite.
We report on numerical results obtained with the matrix $D$
available in the matrix group {\tt  QCD} at the UFL sparse matrix collection as configuration {\tt conf5.4-00l8x8-2000.mtx}
with $\kappa_c = 0.15717$. $Q$ is then given as $Q=P(I - \frac{4}{3}\kappa_c D)$,
with $P$ the permutation
\[
P = I_3 \otimes \left( \begin{array}{cccc} 0 & 0 & 1 & 0 \\
                                           0 & 0 & 0 & 1 \\
					   1 & 0 & 0 & 0 \\
					   0 & 1 & 0 & 0
	\end{array}
	\right) \otimes I_{\frac{n}{12}}.
\]
The dimension
of the system is $n = 12\cdot 8^4 \approx 50\, 000$. We compute $\sign(Q)b$
for a randomly generated vector $b$. To this purpose, we first
compute two numbers $0<a_1<a_2$ such that $\spec(Q) \subset [-a_2,-a_1]
\cup [a_1,a_2]$. We then approximate $\sign(t)$ on $[-a_2,-a_1]
\cup [a_1,a_2]$ with the Zolotarev rational approximation,
see \cite{Petrushev.Popov.book87}. It has the form
$\widehat g(t) = \sum_{i=1}^s \omega_i \frac{t}{t^2+\alpha_i},
\omega_i, \alpha_i > 0$ and it is an $\ell_\infty$-best approximation.
The number of
poles $s$ was chosen such that the $\ell_\infty$-error was less
than $10^{-7}$, that is $s=11$.   To compute $\widehat g(Q)b$, we actually computed
$g(Q) \cdot (Qb)$ with
\begin{equation} \label{sqrt:eq}
g(t) = \sum_{i=1}^s \omega_i \frac{1}{t^2+\alpha_i} \, .
\end{equation}
Since $Q^2$ is hermitian  and positive definite, Algorithm~\ref{LanczosPlusEstimates:alg} 
will produce lower and upper bounds for the exact error.
In our computations we used a deflation technique common in realistic QCD computations  \cite{vdEetal02}: 
We precompute the first, $\lambda_1,\ldots,\lambda_q$ say, eigenvalues of smallest modulus. 
With $\Pi$ denoting the orthogonal projection onto the space spanned by the corresponding
eigenvectors, we then have $\sign(Q)b = \sign(Q)(I-\Pi)b + \sign(Q)(\Pi b)$. Here we know
$\sign(Q)(\Pi b)$ explicitly, so that we now just have to approximate $\sign(Q)(I-\Pi)b$.
In this manner, we effectively shrink the eigenvalue intervals for $Q$,
so that we need fewer poles for an accurate Zolotarev approximation and,
in addition, the linear systems to be solved converge more rapidly.
In QCD practice, this approach results in a major speedup,
since $\sign(Q)b$ must usually be computed repeatedly for various vectors $b$. For Algorithm~\ref{LanczosPlusEstimates:alg} it has the additional advantage that we immediately
have a very good value for $a$, the lower bound on the smallest eigenvalue of $Q^2$ for which we can take $\lambda_q^2$.

Figure~\ref{qcd_8:fig} shows the results that we obtain deflating $q=30$ eigenvalues. The (effective) condition number of the (deflated) matrix $Q^2$ is approximately $50,000$. As in our first example,
the top row reports upper and lower bounds from Algorithm~\ref{LanczosPlusEstimates:alg} whereas the bottom row gives the estimates from \cite{FrSi07} which we know to be a lower bound in this case. As before, we see that going from $k=2$ to $10$ results in a significant gain in accuracy. 
 
\begin{figure} 
\centerline{\includegraphics[width=0.48\textwidth]{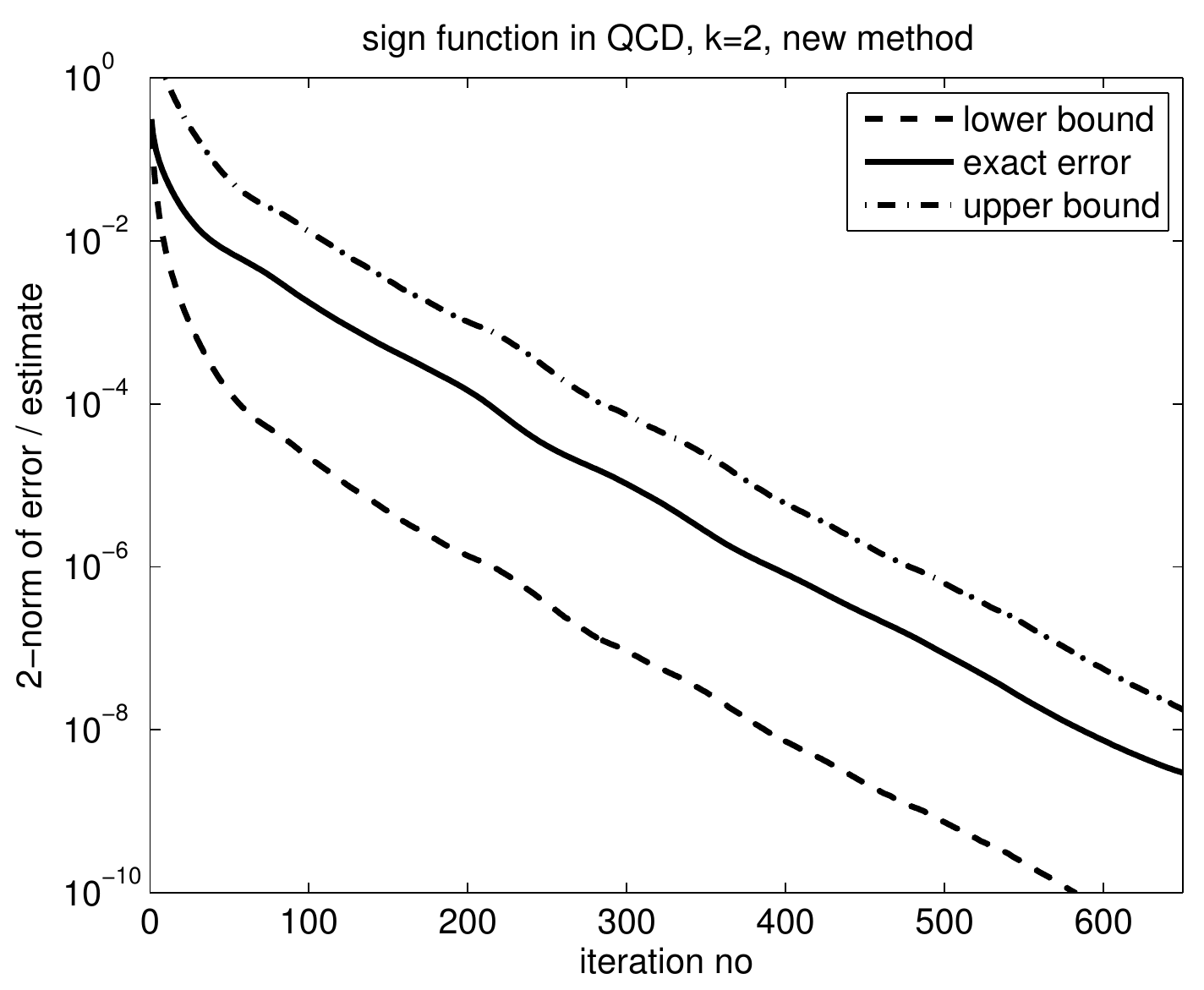}
\hfill \includegraphics[width=0.48\textwidth]{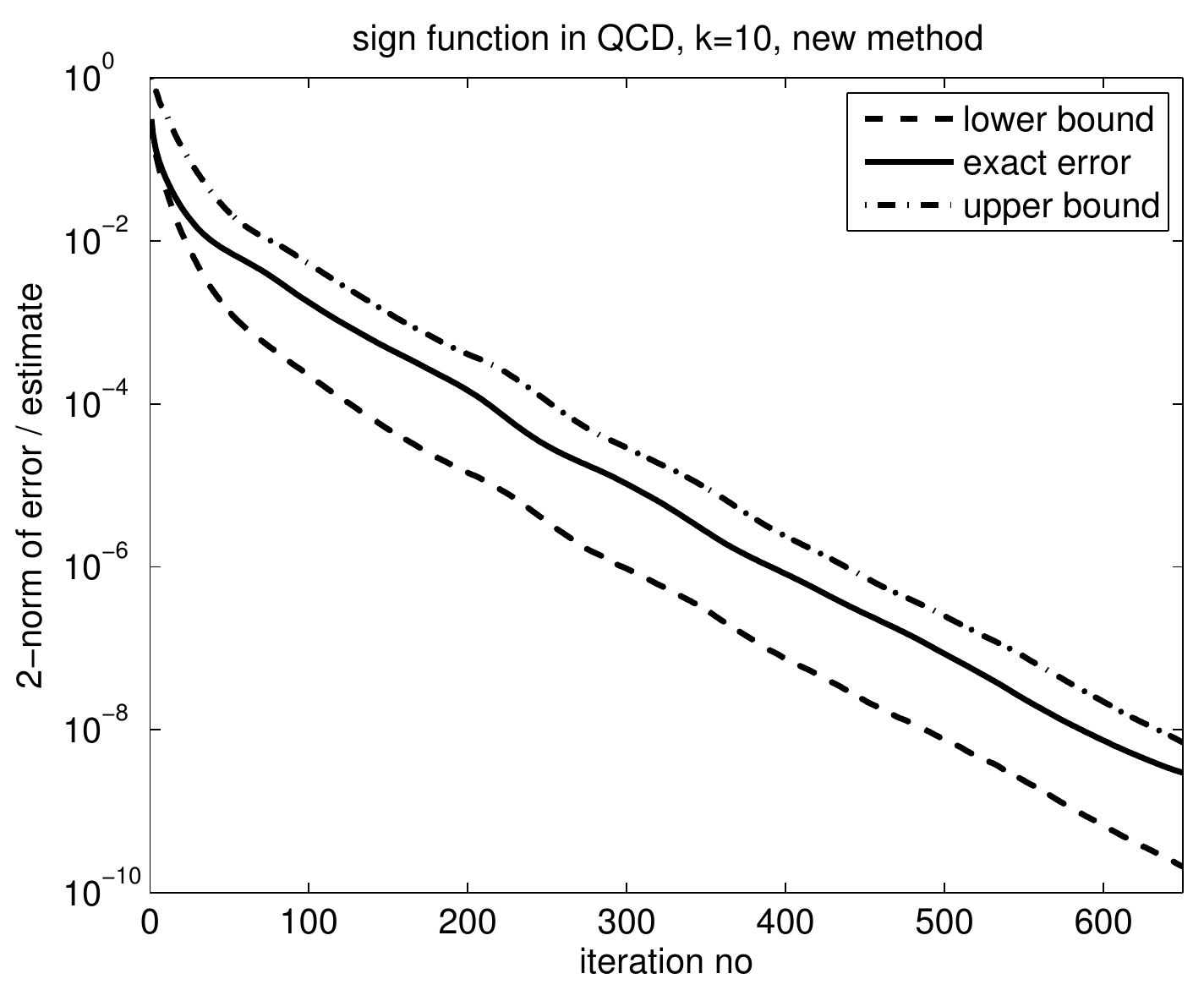}}
\centerline{\includegraphics[width=0.48\textwidth]{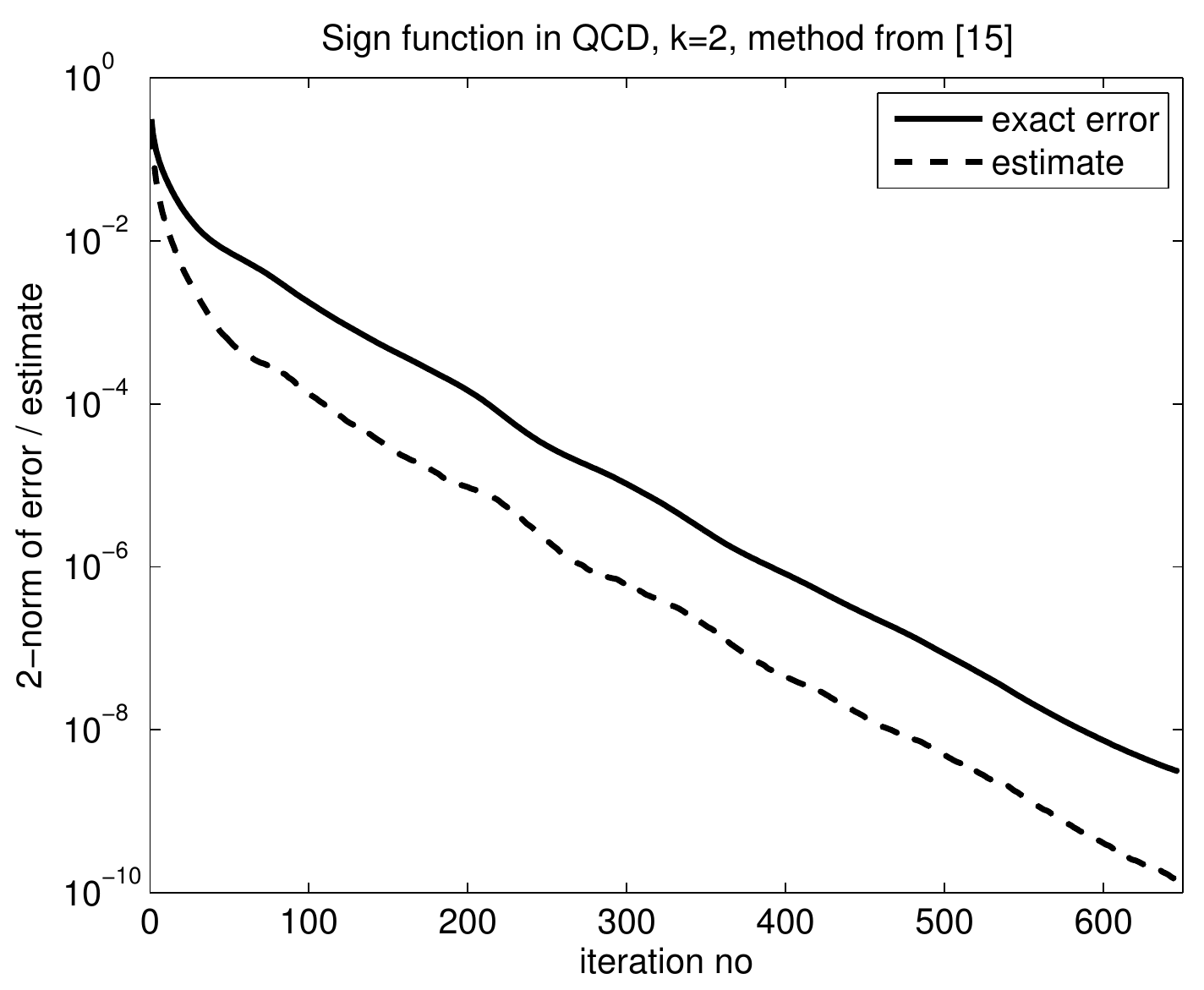}
\hfill \includegraphics[width=0.48\textwidth]{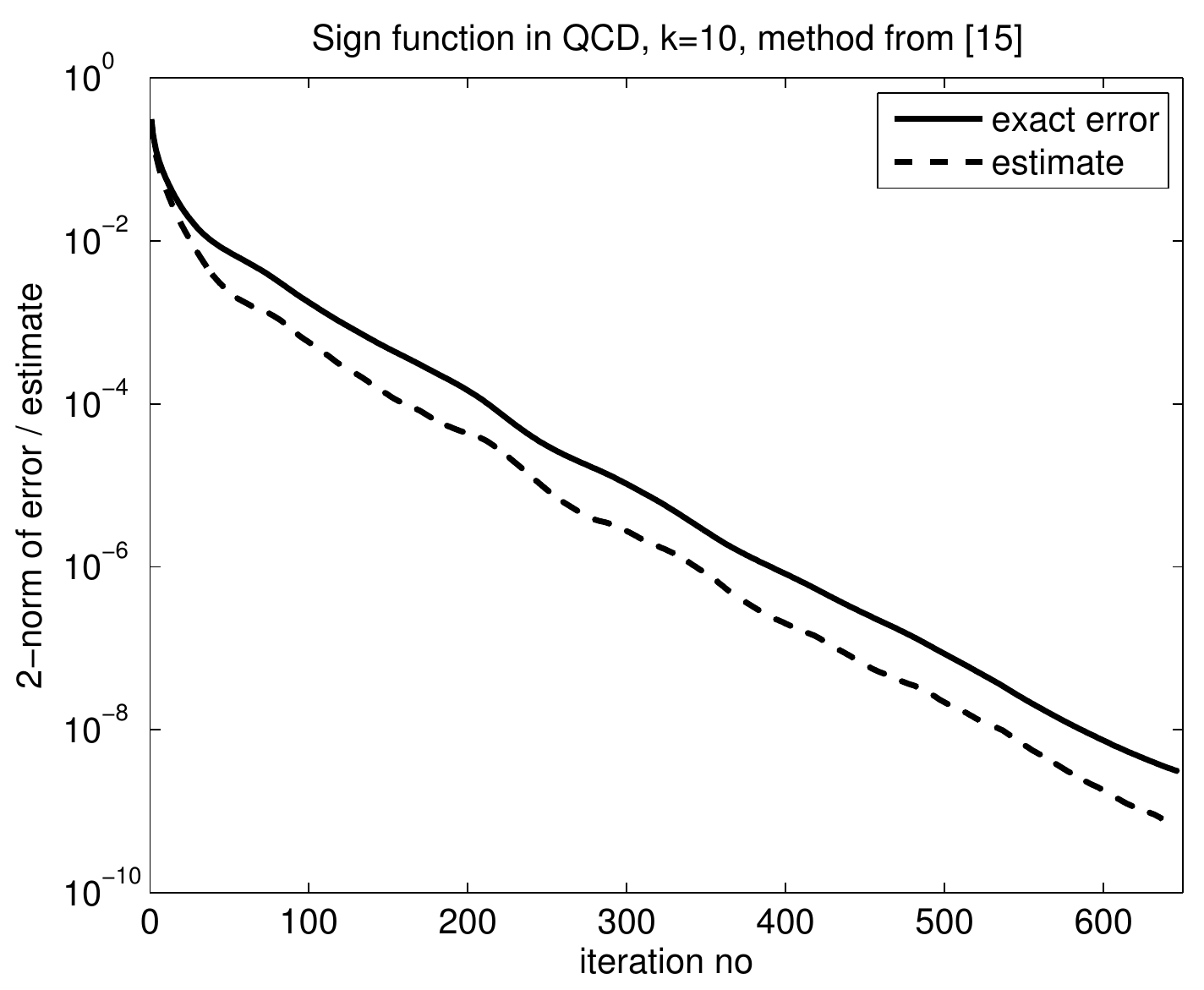}}
\caption{Error bounds and exact error for Zolotarev approximation for $\sign(Q)$ in lattice QCD, 
$8^4$ lattice. Left column: $k=2$, right column: $k=10$. Top row: Algorithm~\ref{LanczosPlusEstimates:alg}, bottom row: method from \cite{FrSi07}\label{qcd_8:fig}}
\end{figure}

Figure~\ref{qcd_16:fig} gives the results for Algorithm~\ref{LanczosPlusEstimates:alg} with $k=10$ for a configuration on a $16^4$ lattice, 
resulting in a matrix $Q$ of dimension $\approx 800,000$. We again deflated the 30 smallest
eigenvalues. The condition number of the deflated matrix $Q^2$ is now approximately $3,700$, i.e., less than 
for the $8^4$ lattice. Therefore, the convergence speed as well as the quality of the bounds are better than 
for the $8^4$ lattice.  

\begin{figure} 
\centerline{\includegraphics[width=0.48\textwidth]{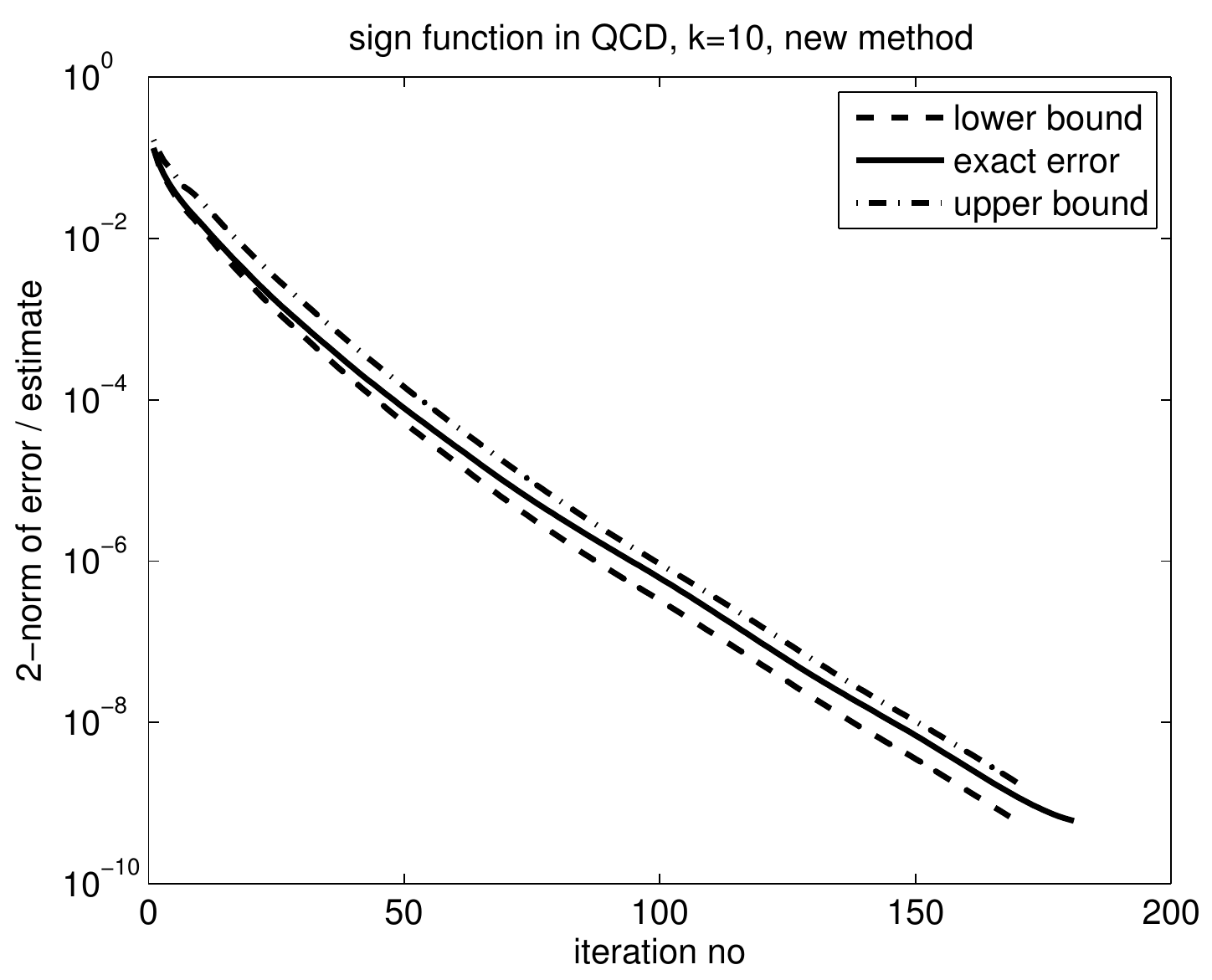}}
\caption{Error bounds and exact error for Zolotarev approximation for $\sign(Q)$ in lattice QCD, 
$16^4$ lattice, Algorithm~\ref{LanczosPlusEstimates:alg}, $k=10$. \label{qcd_16:fig}}
\end{figure}

As a last example we consider the $[10/10]$ Pad\'e approximation to the exponential function. Using $[m/m]$ Pad\'e approximations 
is very common for approximating the matrix exponential; cf.\ \cite{Baker.GravesMorris.96,Lopez.Simoncini.06, Moret.Novati.04b}. Matlab's {\tt expm} uses Pad\'e approximations along with 
the scaling and squaring approach \cite{matlab7}. The partial fraction expansion of the $[10/10]$ Pad\'e approximation to the exponential
has the form 
\[
1 + \sum_{i=1}^5 \frac{\omega_i}{t-\sigma_i} + \frac{\overline{\omega}_i}{t-\overline{\sigma}_i} =: 1 + g(t),
\]
where all the coefficients $\omega_i$ and poles $\sigma_i$ are non-real. We want to compute $(I+g(A))b$, so we
focus on $g(A)b$. Due to the complex poles and coefficients we cannot easily obtain information on the sign of the derivatives of $g_m$, implying that this time we do not know whether Algorithm~\ref{LanczosPlusEstimates:alg} really obtains bounds for the error.

Figure~\ref{exp_2d_laplace:fig} reports the results that we get when computing $g(A)b$ with $A$ the negative discrete Laplacian on a 
$200 \times 200$ grid, $b$ a random vector. Computing $\exp(A)b$ for the negative discrete Laplacian $A$ (or a scalar 
multiple thereof) is a common task when using exponential integrators in semi-discretized parabolic partial differential equations, see \cite{Hairer.Lubich.Wanner.02}.   
 
\begin{figure} 
\centerline{\includegraphics[width=0.48\textwidth]{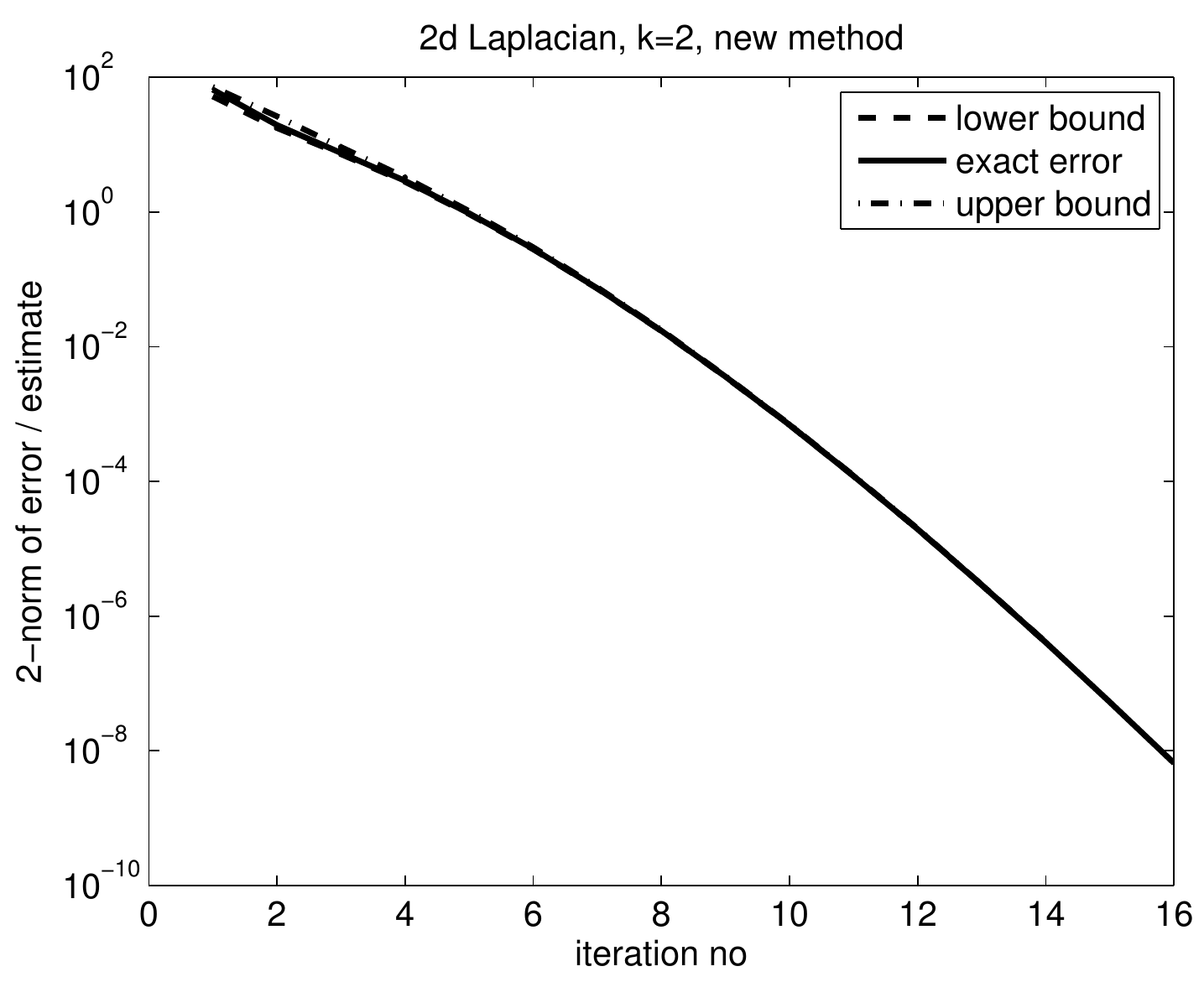}
\hfill \includegraphics[width=0.48\textwidth]{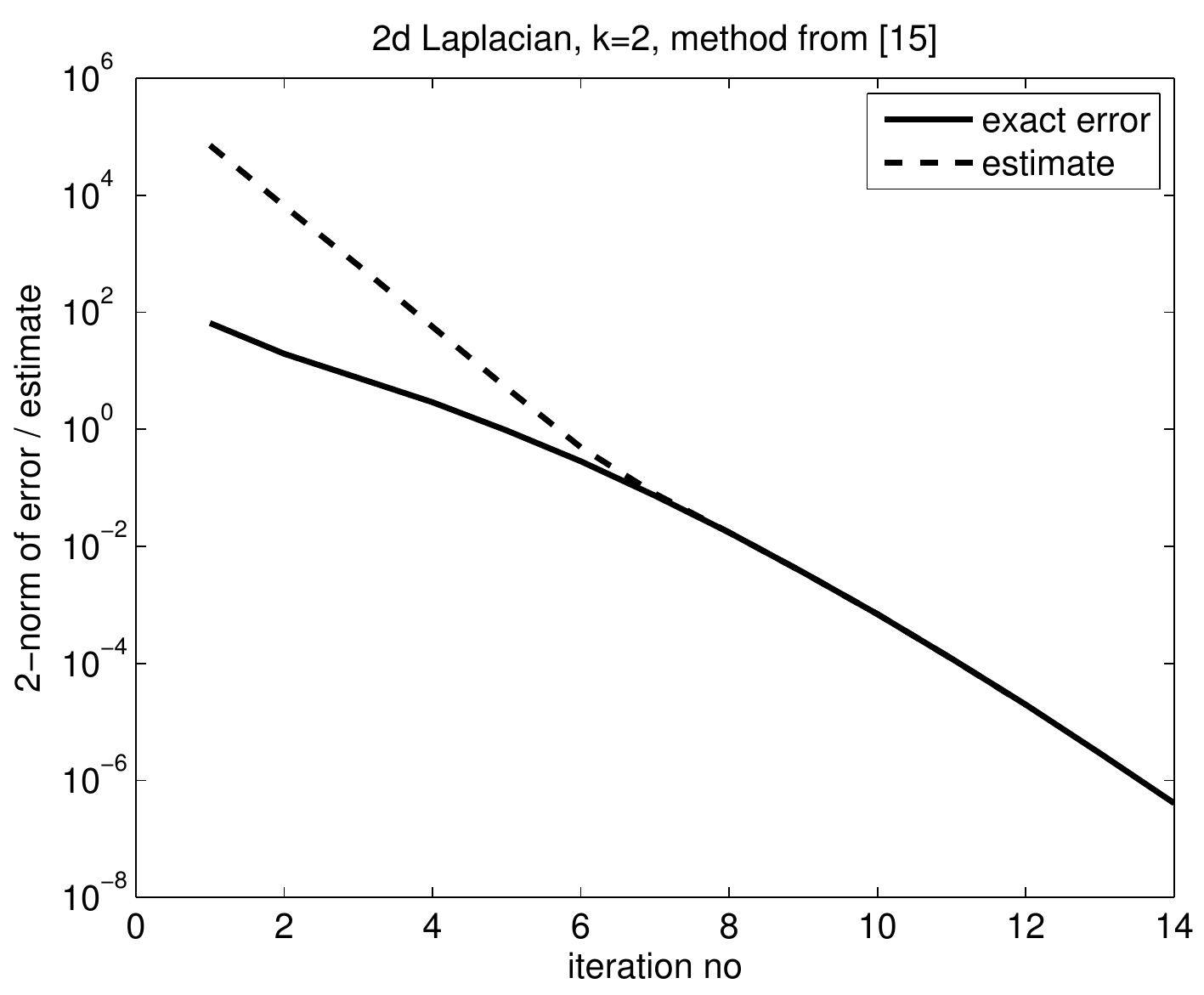}}
\caption{Error bounds and exact error for Pad\'e approximation to the exponential, $A$ negative discrete Laplacian on $200 \times 200$ mesh. \label{exp_2d_laplace:fig}}
\end{figure}
For this example, taking $k=2$ in Algorithm~\ref{LanczosPlusEstimates:alg} is already sufficient to obtain error estimates which are very close 
to the exact error. Although we do not have a theoretical justification, the error estimates produced by
Algorithm~\ref{LanczosPlusEstimates:alg} turn out to indeed represent (tight) lower and upper bounds for the error. The right part of 
Figure~\ref{exp_2d_laplace:fig} shows the error and the error estimates obtained with the approach suggested in \cite{FrSi07}. Note that due 
to the complex shifts, this approach amounts to perform a variant of the CG method which uses the indefinite bilinear form $\langle x, y \rangle_T = 
y^Tx$ on $\mathC^n$. This method thus does not obtain the same iterates as Algorithm~\ref{LanczosPlusEstimates:alg}, but we see that the 
norm of the error is quite similar for both approaches. The error estimate from \cite{FrSi07} is much less precise for the first half of the iterations, whereas for the second half of the iterations it is comparable to the estimates from Algorithm~\ref{LanczosPlusEstimates:alg}. 

The matrix exponential is also used in the analysis of large graphs like those describing social networks. If $A$ is the adjacency matrix 
of such an undirected graph, then $\exp(A)_{ij}$ denotes the {\em communicability} (see \cite{Estrada08}) between nodes $i$ and $j$. Accordingly, $\exp(A)e_i$ gets us the communicabilities of node $i$ with all other nodes. For our numerical experiments we took $i=1$ and we
used the graph {\tt dblp-2010} from group {\tt LAW} of the University of Florida sparse matrix collection. It describes the co-author relation between all authors 
appearing in the DBLP database of journal papers in computer science as of some day in the year 2010. This graph has more than 300,000
nodes and about 1.5 million vertices. Note that $A$ is an indefinite matrix. For this matrix the error estimates from \cite{FrSi07} could not be applied, since the use of the indefinite bilinear form produced breakdowns in the algorithm. In Figure~\ref{dblp:fig} we give only one (the ``lower bound'' $\ell_m$) of the estimates from Algorithm~\ref{LanczosPlusEstimates:alg} for $k=2$ and $k=10$. The ``upper'' bound $u_m$ behaves quite similarly (where we compute $a$ as in the second example). For $k=2$ the estimates appear to systematically represent a lower bound. For $k=10$ we clearly see that the 
estimate does not represent an upper nor lower bound for the error, but we get an estimate for the error which is never more than a 
factor of 5 off the exact error.  
\begin{figure} 
\centerline{\includegraphics[width=0.48\textwidth]{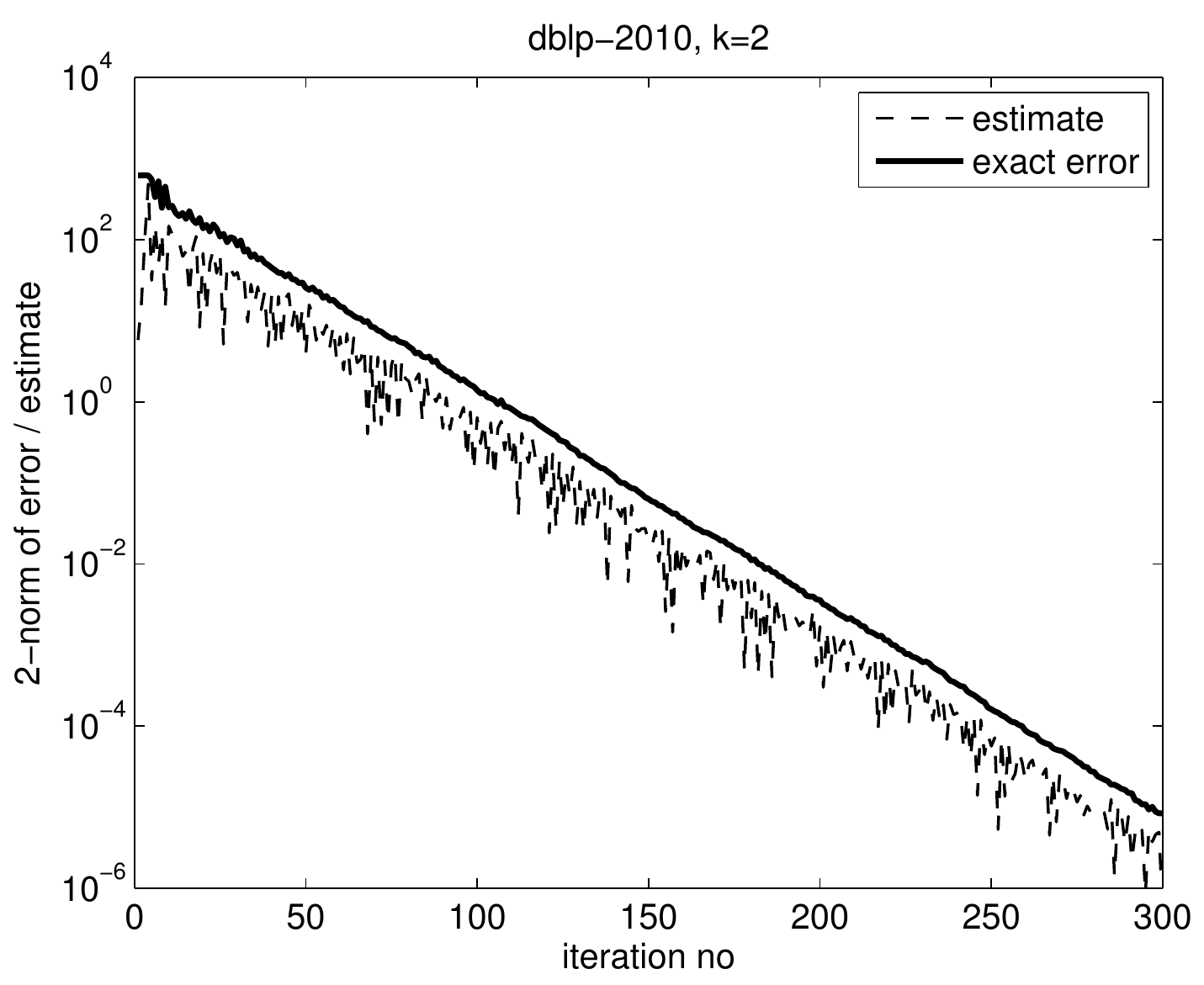}
\hfill \includegraphics[width=0.48\textwidth]{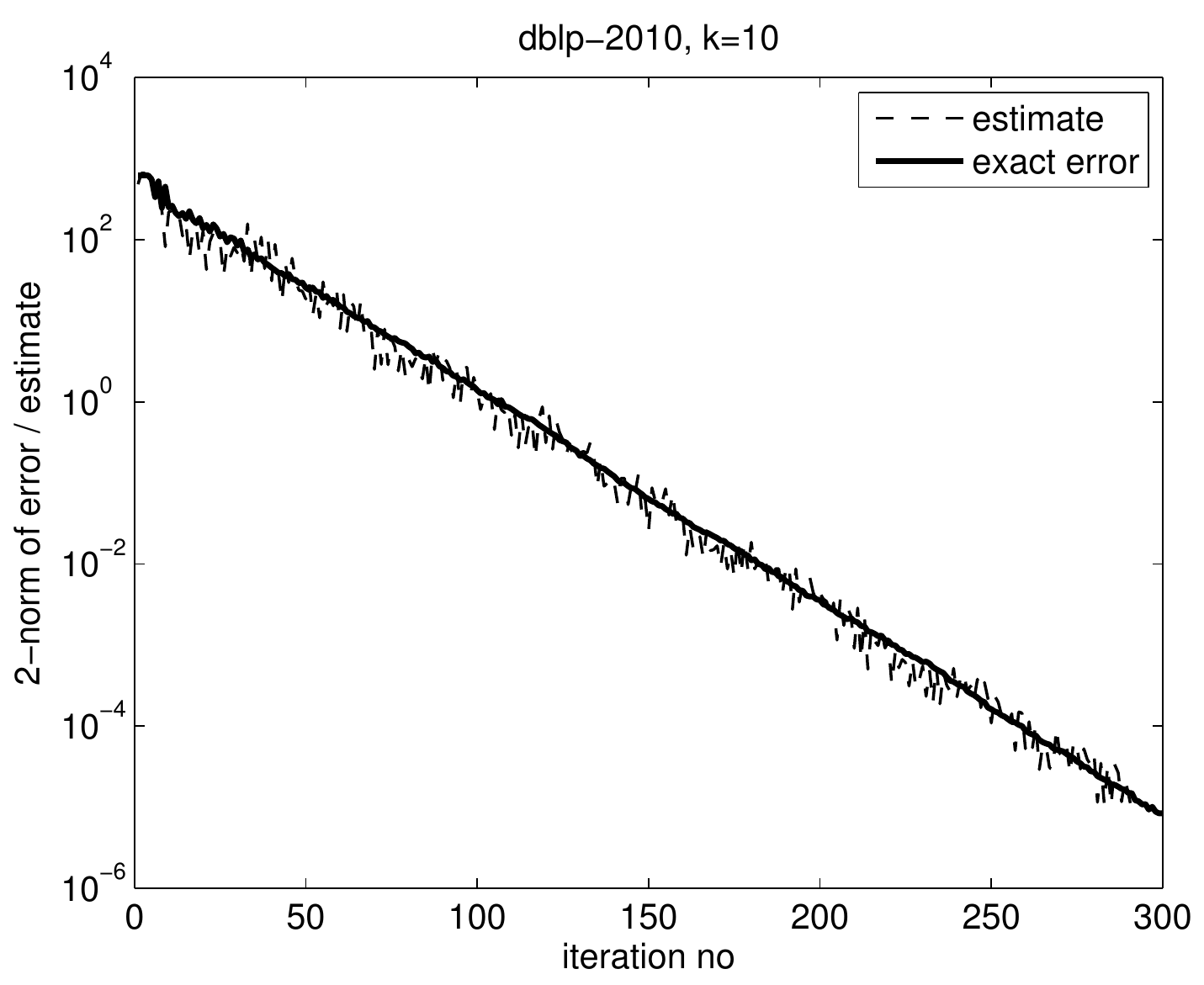}}
\caption{Error bounds and exact error for Pad\'e approximation to the exponential, $A$ adjacency matrix for dblp-2010 graph, $k$ = 2 and $k=10$ in Algorithm~\ref{LanczosPlusEstimates:alg}\label{dblp:fig}}
\end{figure}

\section{Conclusions}
Building on the theory of Golub and Meurant we proposed a novel use of this theory which allows, in particular, to
obtain estimates, lower and upper bounds for the 2-norm of the error of the action of a rational matrix function on a vector. Such estimates are important for rational matrix functions as they can be used as a stopping criterion. Upper bounds have the advantage to provide a {\em reliable} stopping criterion: If we stop the iteration once the upper bound is less than a given threshold $\epsilon$, the exact error is also smaller than $\epsilon$. Our new approach relies on a
secondary, restarted Lanczos process which can be obtained very efficiently at cost which is independent of 
the matrix size. Numerical examples show that the new approach can give very good error bounds with the quality of
the bounds depending on the number $k$ of steps in the secondary Lanczos process and on the condition of the matrix 
function. The effects of rounding errors were not studied, but our approach follows the philosophy put forward  
in \cite{Strakos.Tichy.02,Strakos.Tichy.05} in that it only makes use of ``local orthogonality'', the secondary 
Lanczos process involving just the last $2k$ iterations of the primary Lanczos process. Our approach can, in principle, be extended to the preconditioning idea from \cite{Hochbruck.vdEshof.06}, where instead of $f(A)b$ one computes $r(\tau I -A)^{-1}b$ with $r$ the rational function $r(t) = f(\tau -t^{-1})$, see also \cite{FrSi07,PoSi08}. However, the conditions of Corollary~\ref{bounds:cor} on the signs of the  poles and the coefficients will usually not be fulfilled for $r$, so that we cannot expect to obtain lower and upper bounds. 

\bibliographystyle{siam}

\bibliography{literatur}
  
\end{document}